\documentclass[lettersize,journal]{IEEEtran}
\usepackage{bbm}
\usepackage[english]{babel}
\usepackage[noadjust]{cite}

\usepackage[linesnumbered,ruled,vlined]{algorithm2e}

\ifCLASSINFOpdf
\else
   \usepackage[dvips]{graphicx}
\fi
\usepackage{url}

\usepackage{amsthm}
\usepackage{multirow}
\usepackage{tabularx}
\usepackage{graphicx}
\usepackage{optidef}
\usepackage{amsmath}
\usepackage{bm}
\usepackage{amssymb}
\usepackage{xcolor}
\usepackage{caption}
\usepackage{subcaption}
\usepackage{lipsum}
\usepackage{siunitx}

\usepackage{titlesec}

\DeclareSIUnit{\belmilliwatt}{Bm}
\DeclareSIUnit{\dBm}{\deci\belmilliwatt}

\theoremstyle{definition}
\newtheorem{theorem}{Theorem}

\newtheorem{lemma}{Lemma}

\usepackage{booktabs}       
\usepackage{array}          
\usepackage{threeparttable} 

\newcolumntype{P}[1]{>{\centering\arraybackslash}p{#1}}
\begin{document}
\title{
Constant Modulus Waveform Design with Space-Time Sidelobe Reduction for DFRC Systems
}

\author{
Byunghyun Lee, 
Anindya Bijoy Das, 
David J. Love, 
Christopher G. Brinton, and
James V. Krogmeier,  
\thanks{
A preliminary version of this work was presented at the IEEE International Conference on Communications (ICC), 2024 
\cite{lee2024constant}.
}
\thanks{
This work is supported in part by the National Science Foundation under grants EEC-1941529, CNS-2212565, and CNS-2225578 and the Office of Naval Research under grant N000142112472.
}
\thanks{Byunghyun Lee, David J. Love, Christopher G. Brinton, and James V. Krogmeier are with the Elmore Family School of Electrical and Computer
Engineering, Purdue University, West Lafayette, IN 47907 USA (Email:
\{lee4093, djlove, cgb, jvk\}@purdue.edu).
}
\thanks{Anindya Bijoy Das is with the Department of Electrical and Computer Engineering. The University of Akron, OH 44325 USA (Email: adas@uakron.com).
}
}

\maketitle


\begin{abstract}

Dual-function radar-communication (DFRC) is a key enabler of location-based services for next-generation communication systems.
In this paper, we investigate the problem of designing constant modulus multiple-input multiple-output (MIMO) waveforms for DFRC systems.
We jointly shape the spatial beam pattern and ambiguity function of the transmit space-time matrix to improve target localization accuracy and enhance target resolution in cluttered environments.
For communications, we employ constructive interference (CI)-based precoding, which exploits multi-user and radar-induced interference to enhance MIMO symbol detection.
We develop two novel solution algorithms based on majorization-minimization (MM) and the linearized alternating direction method of multipliers (LADMM) principles.
For the MM approach, we introduce a novel diagonal majorizer for complex quadratic functions, yielding a tighter surrogate and faster convergence than standard largest-eigenvalue-based surrogates.
After majorization, we decompose the approximated problem into independent subproblems that can be efficiently solved via parallelizable coordinate descent.
 To accommodate large MIMO dimensions, we further develop a low-complexity LADMM solution.
We combine a biconvex reformulation and first-order proximal approximations to handle the nonconvex quartic objective without requiring costly matrix inversions.
We evaluate the performance of the proposed algorithms in comparison to the existing DFRC algorithm.
Simulation results demonstrate that the proposed algorithms can substantially enhance target detection and imaging performance due to the reduction of space-time sidelobes.

\end{abstract}

\vspace{-1.5mm}

\IEEEpeerreviewmaketitle

\section{Introduction}

{

In the upcoming 6G era, communication and sensing are expected to seamlessly merge within wireless networks, benefiting both functions with improved utility, spectral, and energy efficiency \cite{brinton2024key6g}.
To support this new trend, called integrated sensing and communications (ISAC), standardization bodies such as the Third Generation Partnership Project (3GPP) 
have initiated study items in ISAC for location-based services such as autonomous driving, intelligent factories, and military surveillance \cite{kaushik2024toward}.
ISAC evolved from spectrum sharing between radar and communications 
to a tighter integration, known as dual-function radar-communication (DFRC), which shares both spectrum and hardware \cite{liu2018toward,liuCramerRaoBoundOptimization2022,yuPrecodingApproachDualFunctional2022}.
 
In this paper, we address the problem of designing a constant modulus probing signal for detecting radar targets while conveying information bits to communication receivers simultaneously.
In radar-centric DFRC systems, high transmit power is typically required to enable precise and reliable target detection and parameter estimation.
However, waveforms with a high peak-to-average power ratio (PAPR) force high-power amplifiers (HPAs) to operate in the nonlinear region, creating signal distortions that destroy the waveform's intended spatial and temporal properties \cite{levanon2004radar,10547068} or requiring significant power back-off.
Therefore, it is crucial to design constant modulus waveforms to maintain the efficiency of HPAs and prevent such distortion.
Several DFRC works have investigated the problem of designing constant modulus waveforms 
\cite{liu2018toward,wu2025quantized,liuDualFunctionalRadarCommunicationWaveform2021,liu2022joint}, while \cite{bazziIntegratedSensingCommunication2023} considered explicit PAPR constraints.

Much existing work on DFRC probing signal design has focused on spatial beam pattern shaping to obtain a strong target response by concentrating energy in the search directions while suppressing sidelobes in undesired directions \cite{lee2024spatial,liuJointTransmitBeamforming2020a,liuDualFunctionalRadarCommunicationWaveform2021,liu2018mu,liu2022joint}.
While beam shaping remains crucial, the waveform's ambiguity characteristics over a coherent processing interval (CPI) are equally critical as they determine angle-delay resolution and interference in cluttered and dense scenes (e.g., low-altitude drone tracking, automotive radar).
This is particularly important in multiple-input multiple-output (MIMO) radar, where target echoes from multiple simultaneous beams may interfere with each other unless their correlations are carefully controlled.

To shape the waveform's correlation profile, similarity metrics have been widely adopted \cite{liu2018toward,bazziIntegratedSensingCommunication2023,qian2018joint,liu2022joint}.
The similarity metric approach forces the designed waveform to remain close to a reference waveform (e.g., chirp), constraining the designed signal to preserve the reference's desirable correlation properties.
While convenient, this approach is inherently suboptimal and offers limited control over correlations at specific angle or range cells due to the dependency on the reference waveform.
Direct correlation optimization approaches address this limitation by minimizing the integrated or peak sidelobe level
\cite{liuJointTransmitBeamforming2020a,liuRangeSidelobeReduction2020,dokhanchiAdaptiveWaveformDesign2021,yuIntegratedWaveformDesign2022,liMIMOOFDMISACWaveform2025}, yet most works ignore the spatial beam pattern aspect.
To capture both spatial and correlation aspects, past radar works \cite{stoica2007probing,wangDesignConstantModulus2019,wang2012design} have studied a trade-off design between beam pattern and correlation.
Despite the benefit of flexibly tuning space-time correlation levels, its application has been limited in the DFRC context.


From a communication perspective, DFRC systems must cope with strong interference from both radar transmissions and spatial multiplexing.
Existing DFRC works treat the radar signal as detrimental interference to be suppressed \cite{bazziIntegratedSensingCommunication2023 ,liu2018toward,liuRangeSidelobeReduction2020}.
However, it is challenging to suppress such large radar signals, particularly in radar-centric DFRC systems with high transmit power.
In this context, constructive interference (CI) precoding has emerged as a promising alternative.
Unlike traditional precoding schemes, CI-based precoding reshapes interference so that it contributes constructively to the communication symbol energy \cite{li2020tutorial}.
Some DFRC works have employed CI-based precoding in probing signal design, focusing on matching the desired beam pattern  \cite{liuDualFunctionalRadarCommunicationWaveform2021,wu2025quantized} and meeting the minimum radiation power toward targets \cite{wangSymbolScalingBasedInterference2025,wangISACEnhancementInterference2025}.
However, these works primarily consider waveform's spatial aspects without accounting for waveform's correlation properties.

To overcome the limitations of the existing works, we design constant-modulus probing signals for radar-centric DFRC systems.
We jointly optimize the spatial beam pattern and the space-time correlations based on the \textit{MIMO ambiguity function} \cite{duly2013time,san2007mimo}.
For communications, we employ CI-based precoding to embed information bits into dual-function signals and efficiently enhance communication symbol energy by leveraging distortion due to radar transmission and multi-user interference.
The resulting optimization is challenging due to its nonconvexity and high-dimensional optimization over the transmit MIMO space-time matrix for space-time sidelobe shaping, whose dimension scales with the array size and the CPI length.
To tackle such challenges, we develop two efficient solution algorithms based on majorization-minimization (MM) and the linearized alternating direction method of multipliers (LADMM). 
Compared to our preliminary work \cite{lee2024constant}, we introduce an additional LADMM algorithm designed for large-scale scenarios, which provides substantial complexity benefits with a modest performance trade-off.
Additionally, we further improve our MM algorithm via parallelization while preserving its monotonic descent property. 
Our contributions can be summarized as follows.
\begin{itemize}
    
    \item 
    We formulate a joint beam shaping and space-time sidelobe suppression problem under a constant modulus constraint for DFRC systems.
    For communication, we employ CI-based precoding to leverage CI from multiuser and radar transmission and improve the sensing-communication trade-off by permitting interference in the constructive direction.
    
    \item 
    We develop an MM-based solution that transforms the original problem into a set of independent linear subproblems with constant modulus constraints.
    We propose a novel surrogate function for quadratic objectives associated with a Hermitian matrix, which outperforms conventional largest eigenvalue-based surrogate functions.
    The approximated problem satisfies strong duality, and thus we alternatively solve the dual problem via a parallelizable coordinate-descent method.
 This algorithm guarantees a monotonically decreasing objective function and inherently provides parallel processing capabilities.
    
\item We develop a low-complexity LADMM algorithm to handle larger matrix dimensions.
By using the variable splitting technique, we decompose the formulated problem into multiple tractable subproblems that admit simple closed-form solutions.
To address the quartic objective, we employ a biconvex formulation combined with proximal updates, thereby eliminating the need for costly matrix inversions. 
Additionally, we exploit the fast Fourier transform (FFT) to accelerate gradient computations.
      
    \item
    Finally, we conduct a series of numerical simulations to evaluate the proposed algorithms and verify their effectiveness in comparison to the existing method \cite{liuDualFunctionalRadarCommunicationWaveform2021}.
    Specifically, we assess the detection and imaging performance of the proposed waveforms.
    
\end{itemize}

The rest of the paper is organized as follows. 
In Sec. II, we provide the system model including the radar and communication models, and formalize our waveform design problem.
Then, in Sec. III and Sec. IV, we develop our MM and LADMM solutions, respectively.
In Sec. V, we evaluate the performances of our proposed algorithms in comparison with the baseline algorithm, and finally, we conclude the paper in Sec. VI.

\textbf{Notation}:
Vectors and matrices are denoted by boldface lowercase and uppercase letters, respectively.
$(\cdot)^T$, $(\cdot)^*$, $(\cdot)^H$, and $(\cdot)^{-1}$ are the transpose, conjugate, conjugate transpose, and inverse operators, respectively.
$|\cdot|$ and $\Vert \cdot\Vert$ denotes the absolute and 2-norm operators, respectively.
$\text{diag}(\cdot)$ is the diagonal matrix, with diagonal entries consisting of the input vector.
$\text{vec}(\cdot)$ is the vectorization of a matrix, while
$\text{mat}_{}(\cdot)$ reshapes a vector into a matrix.
$\text{Tr}(\cdot)$ is the trace of a matrix.
$\mathbb{E}[\cdot]$ is the expectation operator.
$\mathbb{I}$ is the indicator function.
$\odot$ denotes the Hadamard product.
$\circledast$ denotes the linear convolution operator.
${Q}_{i,j}$ denotes the $(i,j)$th entry of a matrix $\textbf{Q}$.
$\textbf{0}$, $\textbf{1}$, and $\textbf{I}$ represent the all-zeros, all-ones, and identity matrices, respectively.
$\angle$ is the phase of a complex number.
$\otimes$ denotes the Kronecker product.
$\nabla_{}$ denotes the gradient operation.
$[\bm{x},\bm{y}]^{(i)}$ denotes $[\bm{x}^{(i)},\bm{y}^{(i)}]$.

\vspace{0.1 in}
\section{System Model and Problem Formulation}
\subsection{System Setup}
Consider a downlink narrowband DFRC system where a base station (BS) operates as a multi-user MIMO transmitter and collocated MIMO radar simultaneously, as depicted in Fig. \ref{fig:system}.
The BS is equipped with transmit and receive arrays of $N_T$ and $N_R$ antennas, respectively.
Without loss of generality, we consider a uniform linear array (ULA) for both the transmit and receive arrays.
The primary function of the considered system is radar sensing, while the secondary function is communication.
To accomplish the dual functions of radar and communication, this paper focuses on downlink transmission, where the BS transmits a discrete-time waveform matrix $\textbf{X}\in \mathbb{C}^{N_T \times L}$ in each transmission block.
The waveform matrix $\textbf{X}$ can be seen as a train of subpulses containing communication information.
The $(n,\ell)$th entry $X_{n,\ell}$ of $\textbf{X}$ represents the $\ell$th radar subpulse and $\ell$th discrete-time transmit symbol of $L$ total for the $n$th transmit antenna.

\begin{figure}[!t]
\center{\includegraphics[width=.72\linewidth]{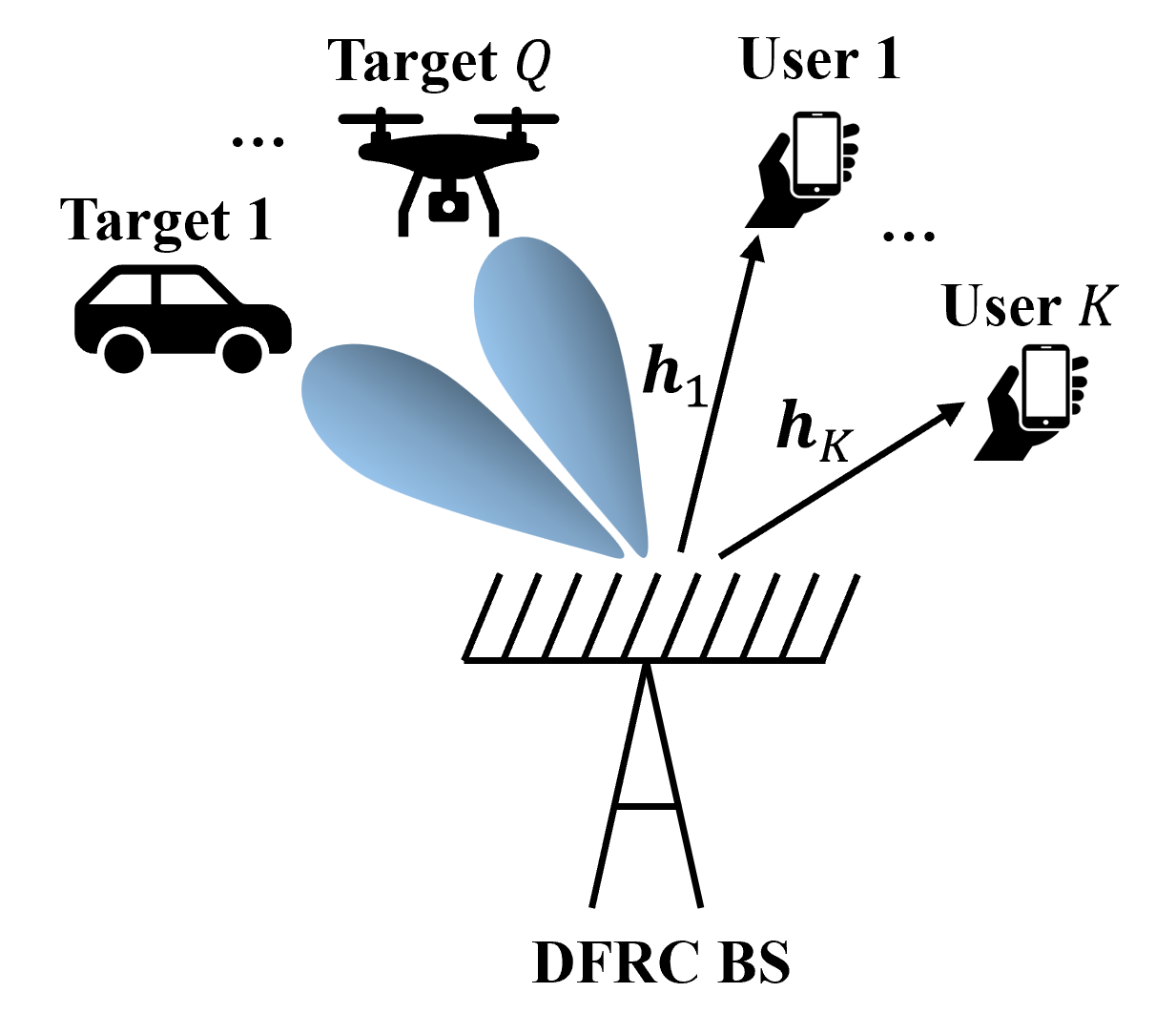}
}
\caption{\small Illustration of a DFRC system. 
}
\label{fig:system}
\end{figure}

{
\subsection{Radar Model}
Consider $Q$ far-field point targets at azimuth angles $\theta_1, \dots, \theta_Q$ and range bins $\tau_1, \dots, \tau_Q$.
To detect the targets, the BS collects reflected signals using the receive antennas.
The received echo signal
at the BS is given by \cite{hua2012receiver,yu2020mimo} \vspace{-2mm}
\begin{equation}
     \textbf{Z}=
    \displaystyle\sum_{q=1}^Q\kappa_q {\textbf{b}(\theta_q)\textbf{a}^H(\theta_q)}{\textbf{X}}\textbf{J}_{\tau_q-\tau_1}+\textbf{W}, \vspace{-2mm}
\end{equation}
where 
$\kappa_q \in\mathbb{C}$ is the complex amplitude proportional to the radar cross-section (RCS) of target $q$, $\textbf{a}(\cdot)\in\mathbb{C}^{N_T}$ is the steering vector of the transmit array, $\textbf{b}(\cdot)\in\mathbb{C}^{N_R}$ is the steering vector of the receive arrays, $\textbf{J}_{\tau_q-\tau_1}\in \mathbb{R}^{L\times L}$ is the shift matrix for target $q$, 
and 
$\textbf{W}\in \mathbb{C}^{N_R \times L}$ is independent and identically distributed (i.i.d.) noise drawn from $\mathcal{CN}(0,\sigma_r^2)$.
The shift matrix accounts for the round-trip delay between the BS and a target, which is given by \cite{horn2012matrix} 
\begin{equation}
    [\textbf{J}_{\tau}]_{i,j} = \begin{cases}
        1, & \text{if }j-i=\tau \\
        0, & \text{otherwise}.
    \end{cases} 
\end{equation}
where $\tau$ is the time shift.

For simplicity, we assume zero-Doppler targets and clutter objects.
Despite this assumption, the extension of our approach to the non-zero Doppler case is straightforward.
}


\subsubsection{Beam Pattern Shaping Cost}

In radar waveform design, it is essential to maximize the mainlobe power directed toward targets while minimizing sidelobes.
This strategy ensures strong return signals from the targets and suppresses undesired signals caused by clutter.
Given the waveform $\textbf{X}$, the beam pattern at angle $\theta$ is given by
${G}(\textbf{X},\theta)= \Vert \textbf{a}^H(\theta)\textbf{X} \Vert ^2=\textbf{a}^H(\theta)\textbf{X}\textbf{X}^H\textbf{a}(\theta)$
, 
where 
$\textbf{a}(\theta)\in\mathbb{C}^{N_T}$ is the steering vector of the transmit array \cite{li2008mimo}.
The beam pattern can be expressed in vector form as 
${{G}}(\bm{x},\theta)
    =\Vert(\textbf{I}_L\otimes\textbf{a}^H(\theta))\bm{x}\Vert^2
    =\bm{x}^H\textbf{A}(\theta)\bm{x}$
where $\textbf{A}(\theta)=\textbf{I}_L \otimes\textbf{a}(\theta)\textbf{a}^H(\theta)$ and $\bm{x}=\text{vec}(\textbf{X})$.
To obtain the desired properties, we minimize the mean square error (MSE) between the ideal beam pattern and the actual beam pattern, which can be expressed as \vspace{-1.5mm}
\begin{equation}
    {g}_{bp}(\alpha,\bm{x})=\displaystyle\sum_{u=1}^U| \alpha G_{d,u}-{G}(\bm{x},\theta_u)|^2,
\end{equation}
where $U$ is the number of angle bins, $\alpha\geq 0$ is the scaling coefficient, and $G_{d,u}$ is the desired beam pattern at angle $\theta_u$.
Here, we have approximated the beam pattern MSE with a finite number $U$ of angle bins.
The scaling coefficient $\alpha$ adjusts the amplitude of the beam pattern that varies according to the BS transmit power.
Given the available closed-form solution to $\alpha$, the beam pattern shaping cost can be expressed in compact vector form as \cite{lee2024constant,liuDualFunctionalRadarCommunicationWaveform2021}
\begin{equation}\label{eq:obj_bp}
  \tilde{g}_{bp}(\bm{x})=\sum_{u=1}^U|\bm{x}^H\textbf{B}_u\bm{x}|^2,  
\end{equation}
where $\textbf{B}_u \triangleq \left({{G_{d,u}\sum\limits_{u'=1}^U\textbf{A}_{}(\theta_{u'})G_{d,u'}}}\right)/{\sum\limits_{u'=1}^UG_{d,u'}^2}-\textbf{A}(\theta_{u})$.


\subsubsection{Space-Time Autocorrelation and Cross-Correlation Integrated Sidelobe Levels (ISLs)}
Since the ambiguity of a radar waveform has a significant impact on parameter estimation quality \cite{duly2013time,san2007mimo}, it is critical to address its ambiguity characteristics.
We consider the space-time correlation function to quantify such ambiguity in the radar waveform.
The space-time correlation function is defined as the correlation between a radar waveform and its echoes reflected from different angles and range bins \cite{wang2012design,wangDesignConstantModulus2019}, which is given by \vspace{-1.5mm}
\begin{equation}\label{eq:space_time_corr}
\chi_{\tau,{q},{q'}}(\textbf{X})=|\textbf{a}^H(\theta_q)\textbf{X}\textbf{J}_{\tau}\textbf{X}^H\textbf{a}(\theta_{q'})|^2.
\end{equation}
{
For a given parameter set $(\tau,q,q')$, the space-time correlation function \eqref{eq:space_time_corr} describes the correlation between angles $\theta_q$ and $\theta_{q'}$ at a range bin $\tau$.
The components $\textbf{a}^H(\theta_{q})\textbf{X}$ and $\textbf{a}^H(\theta_{q'})\textbf{X}$ represent the waveforms radiated toward angles $\theta_{q}$ and $\theta_{q'}$, respectively, and $\textbf{J}_{\tau}$ applies the time delay difference $\tau$ of two return signals.
}
The space-time correlation function can be rewritten in vector form as  
$\chi_{\tau,q,q'}(\bm{x}) = |\bm{x}^H \textbf{D}_{\tau,q,q'}\bm{x}|^2$
where $\textbf{D}_{\tau,q,q'}=\textbf{J}_{-\tau}\otimes \textbf{a}(\theta_{q'})\textbf{a}^H(\theta_{q})$
(See Appendix \ref{sec:appendix_B} for details).
When $q=q'$, the space-time correlation function represents the autocorrelation properties at angle $\theta_q$.
Then, the autocorrelation ISL can be obtained as
\begin{equation}\label{eq:obj_ac}
    \begin{aligned}
         g_{ac}(\bm{x})&=\displaystyle\sum_{q=1}^Q\displaystyle\sum_{\substack{\tau=-P+1, \\ \tau\neq 0}}^{P-1} 
        \chi_{\tau,q,q}(\bm{x}),
    \end{aligned}
\end{equation}
where $Q$ is the number of target directions of interest and $P$ is the maximum round-trip delay\footnote{
\color{black}
The choice of the parameter $P$ is application-specific.
In case when $P=L+1$, all range bins are suppressed, whereas when $P \ll L+1$, partial range bins near zero-delay are suppressed.} of interest with $P-1\leq L$.
On the other hand, when $q\neq q'$, the space-time correlation function $\chi_{\tau,q,q'}$ represents the cross-correlation properties between angles $\theta_q$ and $\theta_{q'}$ at a range bin $\tau$.
The cross-correlation ISL is given by
\begin{equation}\label{eq:obj_cc}
    \begin{aligned}
        g_{cc}(\bm{x})=
    \displaystyle\sum_{q=1}^{Q-1}\displaystyle\sum_{\substack{q'=1, \\ q'\neq q}}^Q\displaystyle\sum_{\tau=-P+1}^{P-1}
    \chi_{\tau,q,q'}(\bm{x}).
    \end{aligned}
\end{equation}
{


}


\subsection{Communication Model and QoS Constraint}

Consider multi-user MIMO transmission where the BS serves $K$ single antenna users simultaneously, i.e., $N_T \geq K$.
We adopt a block-fading channel model where the communication channels remain the same within a transmission block.
The $\ell$th received symbol at user $k$ can be written as 
${y}_{\ell,k}=\textbf{h}_k^H\bm{x}_{\ell}+{n}_{\ell,k},$
where $\bm{x}_{\ell}$ is the $\ell$th column of $\textbf{X}$ containing the $\ell$th communication symbol and $\ell$th radar subpulse, $\textbf{h}_k\in \mathbb{C}^{N_T}$ is the channel from the BS to user $k$, and ${n}_{\ell,k}\in \mathbb{C}$ is Gaussian noise with ${n}_{\ell,k} \sim \mathcal{CN}(0,\sigma^2)$.
We assume the BS has perfect knowledge of the user channels $\textbf{h}_k\in \mathbb{C}^{N_T}$ for $k=1, \dots, K$.
The codeword for user $k$ is given by $\textbf{s}_k=[s_{1,k}, \dots, s_{L,k}]^T\in \mathbb{C}^{L}$ where each symbol $s_{\ell,k}$ is drawn from a constellation $\mathcal{S}$.
{
While data symbols are randomly drawn from a known constellation, the BS knows their specific realizations and uses this knowledge in the precoder design.
}
In what follows, we detail the relationship between the desired codewords $\textbf{s}_1,\textbf{s}_2,\dots,\textbf{s}_K$ and the transmit signal $\textbf{X}$.

\subsubsection*{Per-User Communication QoS Constraint}

To ensure a baseline quality of service (QoS) for the communication users, we consider CI-based precoding to exploit the distortion induced by MU-MIMO and radar transmission.
CI refers to an unintended signal that moves the precoded symbol farther away from its corresponding decision boundaries in the constructive direction.
Unlike conventional precoding that eliminates distortion, CI-based precoding permits interference in the constructive direction, thereby allowing a wider set of feasible solutions.

{
\color{black}
We now derive the CI constraints for given user channels and data symbols, to ensure precoded symbols fall into their respective CI regions.
This paper focuses on the $M$-phase shift keying\footnote{
Although the main focus of this paper is the PSK scenario, it is possible to extend it to quadrature amplitude modulation (QAM), as shown in \cite{li2022practical}.
} (M-PSK) constellation, where $M=4$, i.e., quadrature-PSK (QPSK).
The CI region for each QPSK symbol is determined by the SNR threshold and its boundaries.
The SNR threshold $\gamma_k$ is selected to meet the minimum SNR requirement of the users and defines the distance $|\overrightarrow{OA}|$ between the origin and the desired symbol $A$.
The boundaries of the CI region must be parallel with the decision boundaries and intersect at $A$.
For example, Fig. \ref{fig:CIR} shows the CI region of the QPSK symbol in the first quadrant and its associated boundaries.
The CI direction refers to any direction in which the precoded symbol moves farther away from the boundaries of the CI region.
Vector $\overrightarrow{AC}$ represents interference due to multi-user and radar transmission.
$\overrightarrow{OC}=\overrightarrow{OA}+\overrightarrow{AC}$ corresponds to the noise-less precoded symbol.

From the geometry, it can be observed that $\overrightarrow{OC}$ falls into the CI region if the distortion $\overrightarrow{AC}$ is in the CI direction.
This holds if and only if 
$|\overrightarrow{BD}|\geq|\overrightarrow{BC}|$ where point $B$ is the projection of point $C$ onto the line at an angle $\Lambda$ to the boundaries of the CI region.
}
The length $|\overrightarrow{BC}|$ can be expressed as $|\overrightarrow{BC}|=|\Im\{\textbf{h}^H_k\bm{x}_{\ell}e^{-j\angle s_{\ell,k}}\}|$, while the length $|\overrightarrow{BD}|$ can be obtained as
$|\overrightarrow{BD}|=|\overrightarrow{AB}|\tan \Lambda=\Re\{\textbf{h}^H_k\bm{x}_{\ell}e^{-j\angle s_{\ell,k}}-\sigma \sqrt{\gamma_k}\}\tan \Lambda$.
\color{black}

Combining the above results, the CI constraint for $\ell$th symbol of user $k$ can be formulated as \cite{liInterferenceExploitationPrecoding2018}
\begin{align*}\label{ineq:comm_constraint}
    \Re&\{\textbf{h}^H_k\bm{x}_{\ell}e^{-j\angle s_{\ell,k}}-\sigma \sqrt{\gamma_k}\}\tan \Lambda -|\Im\{\textbf{h}^H_k\bm{x}_{\ell}e^{-j\angle s_{\ell,k}}\}| \geq 0.
\end{align*}
{\color{black}
The above inequality incorporates phase information of data symbols, which is used for nonlinear mapping from $s_{\ell,k}$ to $\bm{x}_{\ell}$. 
}
The above CI constraint can be transformed into \cite{liuDualFunctionalRadarCommunicationWaveform2021}
\begin{equation}\label{ineq:CI}
    \Re\{\hat{\textbf{h}}_{\ell,m}^H\bm{x}_{\ell}\}\geq \Gamma_{m},
    \ \forall \ell=1,2,\dots,L, \ \forall m=1,2,\dots,2K,  
\end{equation}
where 
\begin{align*}
   \hat{\textbf{h}}^H_{\ell,2k}& \triangleq  \textbf{h}_k^He^{-j\angle s_{\ell,k}}(\sin \Lambda -j\cos \Lambda), \\
    \hat{\textbf{h}}^H_{\ell,2k-1}& \triangleq  \textbf{h}_k^He^{-j\angle s_{\ell,k}}(\sin \Lambda +j\cos \Lambda), \; \textrm{and}\\
    \Gamma_{2k}  &\triangleq \sigma \sqrt{\gamma_k}\sin \Lambda,\Gamma_{2k-1} \triangleq \sigma \sqrt{\gamma_k}\sin \Lambda.
\end{align*}

\begin{figure}[!t]
\center{\includegraphics[width=.74\linewidth]{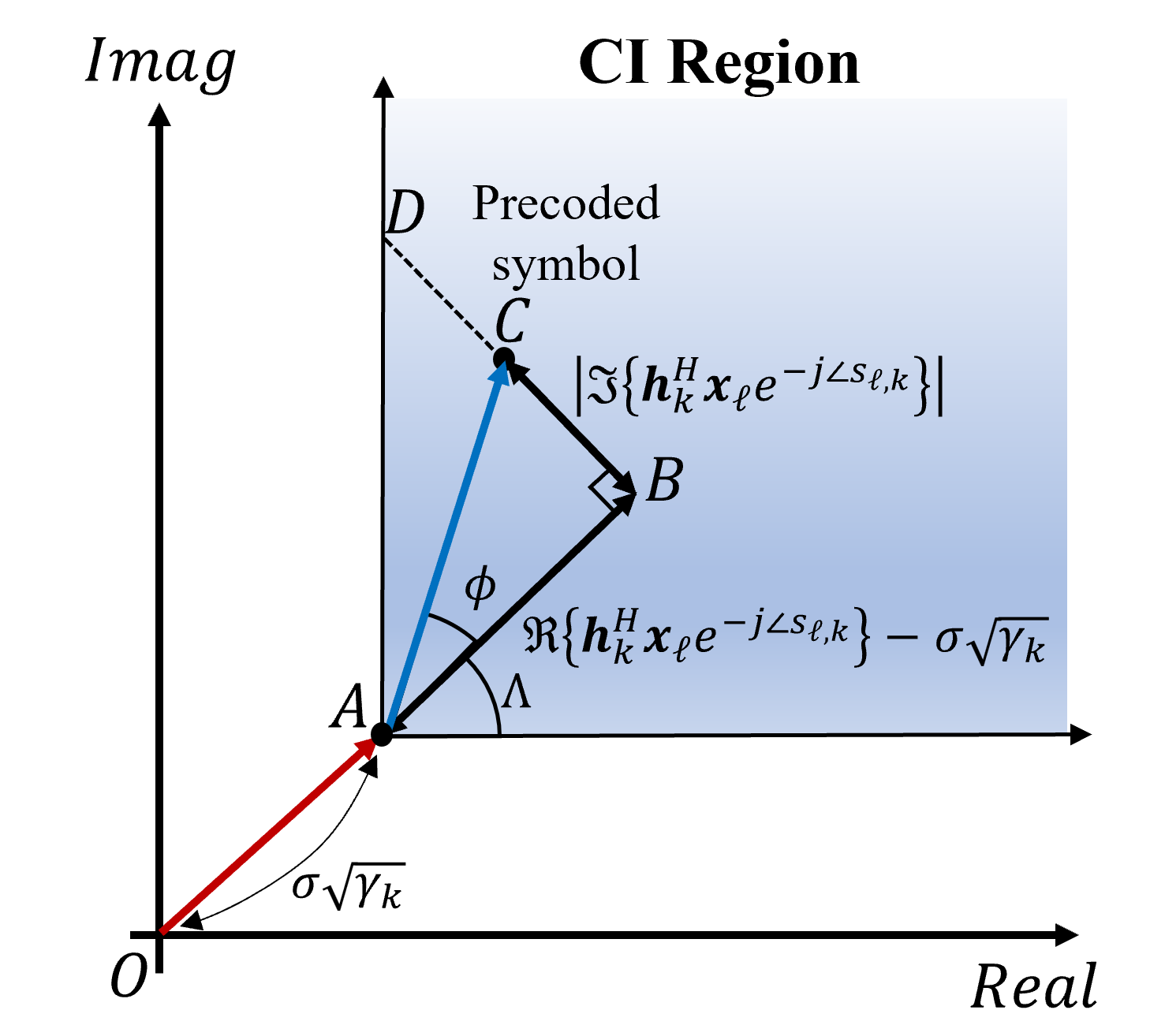}
}
\caption{\small  Constructive interference (CI) region.
The $\ell$th noiseless received symbol $\textbf{h}_{k}^H\bm{x}_{\ell}$ for user $k$ lies within the CI region if the inequality 
$|\overrightarrow{BD}|\geq|\overrightarrow{BC}|$ holds.
}
\label{fig:CIR}
\end{figure}

Due to the limited space, we refer the readers to {\cite{liInterferenceExploitationPrecoding2018,liuDualFunctionalRadarCommunicationWaveform2021}} for a detailed derivation.
With this, the CI constraint can be reformulated with respect to vector $\bm{x}$ as $    \Re\{\tilde{\textbf{h}}^H_{\ell,m}\bm{x}\} \geq \Gamma_m, $
where $\tilde{\textbf{h}}^H_{\ell,m} \triangleq  \textbf{e}^T_{\ell,L} \otimes \hat{\textbf{h}}^H_{\ell,m}$ and $\textbf{e}^T_{\ell,L}$ denotes the $\ell$th column of the $L \times L $ identity matrix.


\subsection{Problem Formulation}
{

Our primary goal is to flexibly design the trade-off between the beam pattern and waveform correlations while meeting the communication users' QoS requirements.
The beam shaping cost function $\tilde{g}_{bp}(\bm{x})$ aims to maximize the strength of a target return while minimizing spatial sidelobes in the undesired directions.
In the meantime, the ISL cost functions control the waveform's overall correlation level, reducing interference between different objects and improving their separation.
For communication, we impose the CI constraint to meet the QoS requirement while relaxing the feasible region relative to interference-suppression constraints.
By combining these sensing and communication design goals, the waveform design problem is formulated as
}
\begin{equation}\label{eq:prob_formulation1}
\begin{aligned}
& \underset{\bm{x}}{\min}
& &  \omega_{bp}\tilde{g}_{bp}(\bm{x})+\omega_{ac} g_{ac}(\bm{x})+\omega_{cc} g_{cc}(\bm{x})  \\ 
& \text{s.t.}
& & \textbf{C1}: \Re\{\tilde{\textbf{h}}^H_{\ell,m}\bm{x}\} \geq \Gamma_m, \forall \ell,m \\
& & & \textbf{C2}: | {x}_{n} | = \sqrt{\frac{P_T}{N_T}}, \ \forall n=1,2,\dots,LN_T \\
\end{aligned}
\end{equation}
where $\omega_{bp}, \omega_{ac}, \omega_{cc} \geq 0$ are the weights for the beam pattern shaping cost \eqref{eq:obj_bp}, autocorrelation ISL \eqref{eq:obj_ac} and cross-correlation ISL \eqref{eq:obj_cc}, respectively. $\textbf{C1}$ is the communication QoS constraint, and $\textbf{C2}$ is the constant modulus constraint.
By normalizing the constant modulus constraint, we can reformulate the above problem as
\begin{equation}\label{eq:prob_formulation2}
\begin{aligned}
& \underset{\bm{x}}{\min}
& &  \omega_{bp}\tilde{g}_{bp}(\bm{x})+\omega_{ac} g_{ac}(\bm{x})+\omega_{cc} g_{cc}(\bm{x})  \\ 
& \text{s.t.}
& & \textbf{C1}: \Re\{\tilde{\textbf{h}}^H_{\ell,m}\bm{x}\} \geq \tilde{\Gamma}_m, \forall \ell,m\\
& & & \textbf{C2}: | {x}_{n} | = 1, \ \forall n=1,2,\dots,LN_T 
\end{aligned}
\end{equation}
where $\tilde{\Gamma}_m=\sqrt{\frac{N_T}{P_T}}{\Gamma_m}$. 
\begin{theorem}\label{theorem:nonconvex}
Problem \eqref{eq:prob_formulation2} is nonconvex.
\end{theorem}
See Appendix \ref{sec:appendix_C} for the proof.

Our per-block design jointly shapes the space-time properties of a constant modulus transmit block, providing improved flexibility and trade-offs over indirect similarity-based schemes \cite{liu2018mu,liuJointTransmitBeamforming2020a,liuDualFunctionalRadarCommunicationWaveform2021,liu2022joint,lee2024spatial}.
This approach is particularly effective in dense sensing environments (e.g.,  automotive radar), where multiple closely spaced targets and strong clutter require fine-grained control of correlations.
The weight parameters can be selected depending on the scenario.
For example, a higher $\omega_{bp}$ value can be used to boost the SNR of a weak target, whereas $\omega_{ac}$ and $\omega_{cc}$ can be emphasized under cluttered environments to suppress interference.

The formulated problem is inherently nonconvex due to the nonconvex quartic objective and constant modulus constraint.
Moreover, optimizing the full space-time matrix becomes challenging at large array sizes and block lengths.
In the subsequent sections, we develop two solution algorithms: an MM method that yields monotone descent and admits parallel updates, and a low-complexity LADMM method that can handle larger dimensions via proximal updates and simple projections.

\vspace{-1mm}




\section{MM Algorithm}\label{sec:MM}

In this section, we introduce our solution that leverages the majorization-minimization (MM) technique and the method of Lagrange multipliers.
We first derive a linear majorizer function of the fourth-order objective in \eqref{eq:prob_formulation2} to handle its nonconvexity.
We propose a novel diagonal majorizer designed for complex Hermitian matrices. 
This approach yields a tighter surrogate than conventional largest-eigenvalue-based methods \cite{liFastAlgorithmsDesigning2018,liuDualFunctionalRadarCommunicationWaveform2021,zhao2016unified,sun2016majorization}, particularly for ill-conditioned matrices, thereby significantly accelerating convergence.
Under this majorization, the problem \eqref{eq:prob_formulation2} can be approximated as a linear program with a constant modulus constraint.
We then decompose this approximation into multiple independent subproblems.
Since strong duality holds for these subproblems, they can be solved efficiently and in parallel using coordinate descent.
In the following, we describe the majorization process of \eqref{eq:prob_formulation2} and the solution based on dual problems.

\subsection{Majorizing with an Improved Majorizer}


To majorize the objective, we begin by rewriting the quadratic term in the beam pattern shaping cost as $\bm{x}^H\textbf{B}_u\bm{x}=\text{Tr}(\bm{x}\bm{x}^H\textbf{B}_u)=\text{vec}^H(\bm{x}\bm{x}^H)\text{vec}(\textbf{B}_u)$ \cite{liFastAlgorithmsDesigning2018}.
Then, following the prevalent approach used in \cite{liFastAlgorithmsDesigning2018,liuDualFunctionalRadarCommunicationWaveform2021,zhao2016unified,sun2016majorization}, the fourth-order beam pattern shaping cost can be expressed as 
$\sum_{u=1}^{U}|\bm{x}^H\textbf{B}_u\bm{x}|^2= \text{vec}^H(\bm{x}\bm{x}^H){\bm{\Psi}}_1\text{vec}(\bm{x}\bm{x}^H),$
where
${\bm{\Psi}}_1\triangleq\sum_{u=1}^{U}\text{vec}(\textbf{B}_u)\text{vec}^H(\textbf{B}_u)$.
It can be verified that ${\bm{\Psi}}_1$ is an $(L^2N_T^2 \times L^2N_T^2)$ Hermitian positive definite matrix.
Following this approach, the objective can be expressed as \vspace{-2mm}
\begin{equation}\label{eq:obj_majorized1}
    \begin{aligned}    g(\bm{x})&=\text{vec}^H(\bm{x}\bm{x}^H)\underbrace{\left(\omega_{bp}\bm{\Psi}_1+\omega_{ac}\bm{\Psi}_2+\omega_{cc}\bm{\Psi}_3\right)}_{\bm{\Psi}}\text{vec}(\bm{x}\bm{x}^H) \\ 
    &=\text{vec}^H(\bm{x}\bm{x}^H)\bm{\Psi}\text{vec}(\bm{x}\bm{x}^H), \vspace{-4mm}
\end{aligned}
\end{equation}
where  \vspace{-1.5mm}
\begin{align*} 
    \bm{\Psi}_2 &\triangleq \displaystyle\sum_{q=1}^Q\displaystyle\sum_{\substack{\tau=-P+1, \\ \tau\neq 0}}^{P-1}\text{vec}(\textbf{D}_{\tau,q,q})\text{vec}^H(\textbf{D}_{\tau,q,q}) ; \; \\
    \textrm{and} \; \; \bm{\Psi}_3 &\triangleq\displaystyle\sum_{q=1}^{Q}\displaystyle\sum_{\substack{q'=1, \\ q'\neq q}}^Q\displaystyle\sum_{\tau=-P+1}^{P-1}\text{vec}(\textbf{D}_{\tau,q,q'})\text{vec}^H(\textbf{D}_{\tau,q,q'}).
\end{align*}
Then, we use the following lemma to construct a majorizer of the fourth-order objective function.
\begin{lemma}\label{lemma:Taylor} 
\cite[(13)]{sun2016majorization}
Let $\textbf{Q},\textbf{R}$ be Hermitian matrices with $\textbf{R}\succeq \textbf{Q}$. 
Then, a quadratic function $\bm{u}^H\textbf{Q}\bm{u}$ can be majorized at a point $\bm{u}_t$ as \vspace{-1.mm}
\begin{align*}\label{ineq:majorizer}
    \bm{u}^H\textbf{Q}\bm{u} \leq \bm{u}^H\textbf{R}\bm{u}&+2\Re\{\bm{u}^H(\textbf{Q}-\textbf{R})\bm{u}_t\} +\bm{u}_t^H(\textbf{R}-\textbf{Q})\bm{u}_t. \vspace{-2mm}
\end{align*}
\end{lemma}
According to the above lemma, we can majorize a quadratic function by choosing a matrix $\textbf{R}$ such that $\textbf{R}\succeq \textbf{Q}$.
To simplify the right-hand side, matrix $\textbf{R}$ is required to be diagonal \cite{song2015sequence}.
In the literature, the predominant choice for $\textbf{R}$ is $\textbf{R}=\lambda_Q \textbf{I}$ where $\lambda_Q$ is the largest eigenvalue of $\textbf{Q}$ \cite{liFastAlgorithmsDesigning2018,liuDualFunctionalRadarCommunicationWaveform2021,zhao2016unified,sun2016majorization}.
However, this majorizer can be loose when $\textbf{Q}$ is ill-conditioned.
\cite{song2015sequence} proposed a novel diagonal matrix structure to enable tight majorization for the case where $\textbf{Q}$ is a non-negative symmetric matrix, which is not directly applicable to \eqref{eq:obj_majorized1} since $\bm{\Psi}$ is complex Hermitian.
Here, we develop a more general majorizer for any quadratic function with a complex Hermitian matrix based on the following lemma.




\begin{table}[t]
\centering
\caption{Comparison of diagonal matrices}
\label{tab:maj-slim}
\renewcommand{\arraystretch}{1.2} 
\begin{tabular}{l c p{3.5cm}}
\toprule
\textbf{Majorizer} & \textbf{Conditions} & \textbf{Characteristics} \\
\midrule
$\lambda_{\max}(\mathbf{Q})\mathbf{I}$ & $\mathbf{Q} \in \mathbb{C}^N$ & Universal, may be loose when $\textbf{Q}$ is ill-conditioned. \\

$\operatorname{diag}(\mathbf{Q}\mathbf{1})$ & $\mathbf{Q} \in \mathbb{S}_+^N$ & Tighter, limited to real non-negative symmetric matrices. \\

{$\operatorname{diag}(\hat{\mathbf{Q}}\mathbf{1})$} { (Ours)} & $\mathbf{Q} \in \mathbb{H}^N$ & Tighter, generalizes to all Hermitian matrices. \\
\bottomrule
\end{tabular}
\end{table}


\begin{lemma}\label{lemma:diag}
Let $\textbf{Q}$ be a Hermitian matrix.
Let $\hat{\textbf{Q}}$ be a $N \times N$ matrix such that $\hat{{Q}}_{i,j}=|{Q}_{i,j}|$.{
\color{black}
Then, $\text{diag}(\hat{\textbf{Q}}\textbf{1}_{N\times 1})\succeq \textbf{Q}$.
}
\end{lemma}
\begin{proof}
    For any $\bm{u}$, we have
\begin{equation}
    \begin{aligned}
    &\bm{u}^H(\text{diag}(\hat{\textbf{Q}}\textbf{1})-\textbf{Q})\bm{u}
    = \displaystyle\sum_{i,j}|{Q}_{i,j}||u_i|^2-\displaystyle\sum_{i,j}u_i^*{Q}_{i,j}u_j
    \\ \nonumber
    = &
    \frac{1}{2}\displaystyle\sum_{i,j}\left(2|{Q}_{i,j}||u_i|^2-2\Re\{{Q}_{i,j}u_i^*u_j\}\right)\\
    = &
    \frac{1}{2}\displaystyle\sum_{i,j}\left(|{Q}_{i,j}||u_i|^2+|{Q}_{j,i}||u_j|^2
    - 2\Re\{{Q}_{i,j}u_i^*u_j\}\right) \\ \nonumber
    =&\frac{1}{2}\displaystyle\sum_{i,j}\left(|{Q}_{i,j}||u_i|^2+|{Q}_{i,j}|||u_j|^2 - 2\Re\{{Q}_{i,j}u_i^*u_j\}\right),
\end{aligned}
\end{equation}
where the last equality follows from $|Q_{i,j}|=|Q_{j,i}^*|=|Q_{j,i}|$.{
\color{black}
Note we drop the size of the all-ones vector for ease of notation.
}
Now, for any $i,j$, we have $|{Q}_{i,j}||u_i|^2 +|{Q}_{i,j}||u_j|^2 
    - 2\Re\{{Q}_{i,j}u_i^*u_j\} 
  \geq   |{Q}_{i,j}|(|u_i|- |u_j|)^2 \geq 0$,
which follows from the fact that $|{Q}_{i,j}||u_i||u_j|\geq \Re\{{Q}_{i,j}u_i^*u_j\}$.
It follows that $\bm{u}^H(\text{diag}(\hat{\textbf{Q}}\textbf{1})-\textbf{Q})\bm{u} \geq 0$.
\end{proof}

}

{

Let $\mathbb{H}^{N\times N}$ and $\mathbb{S}_+^{N\times N}$ be the set of $N\times N$ Hermitian and real nonnegative symmetric matrices, respectively.
Lemma \ref{lemma:diag} shows that our proposed majorizer $\text{diag}(\hat{\textbf{Q}}\textbf{1})$ applies to any quadratic function $\bm{x}^H\textbf{Q}\bm{x}$ with $\textbf{Q}\in \mathbb{H}^{N\times N}$, whereas the majorizer in \cite{song2015sequence} applies only to the case when $\textbf{Q}\in \mathbb{S}_+^{N\times N}$.
From a signal processing perspective, this provides a significant advantage because the proposed diagonal matrix can replace the standard largest eigenvalue-based majorizers, $\lambda_Q \textbf{I}$, in many existing MM algorithms  \cite{liFastAlgorithmsDesigning2018,liuDualFunctionalRadarCommunicationWaveform2021,zhao2016unified,sun2016majorization}, accelerating their convergence.
A comparison of ours and other existing majorizers is provided in Table \ref{tab:maj-slim}.
Using Lemma \ref{lemma:diag}, a tight majorizer for the beam shaping cost can be constructed as follows (with the proof in \cite{lee2024constant}).
\begin{lemma}
\label{lemma:quad_maj}
    Let $\hat{\bm{\Psi}}$ be a matrix such that $\hat{{\Psi}}_{i,j}=|{{\Psi}}_{i,j}|$ for all $i,j$.
    The objective function \eqref{eq:obj_majorized1} can be majorized as
    \begin{equation}\label{ineq:obj_quad}
        g(\bm{x})\leq \bm{x}^H\bm{\Phi}\bm{x} + const, 
        \vspace{-3mm}
    \end{equation}
    where 
 \begin{align*}    
 {\bm{\Phi}}&\triangleq 2\left(\omega_{bp}\bm{\Phi}_1 + \omega_{ac}\bm{\Phi}_2 +\omega_{cc}\bm{\Phi}_3
-\left({\textbf{E}}\odot\bm{x}_t\bm{x}_t^H\right)\right), \\
\bm{\Phi}_1 & \triangleq \displaystyle\sum_{u=1}^U\bm{x}_t^H{\textbf{B}}_u^H\bm{x}_t{\textbf{B}}_u, 
\bm{\Phi}_2  \triangleq \displaystyle\sum_{q=1}^Q\displaystyle\sum_{\substack{\tau=-P+1, \\ \tau\neq 0}}^{P-1}\bm{x}_t^H\textbf{D}^H_{\tau,q,q}\bm{x}_t\textbf{D}_{\tau,q,q},
\\
\bm{\Phi}_3 & \triangleq \displaystyle\sum_{q=1}^{Q}\displaystyle\sum_{\substack{q'=1, \\ q'\neq q}}^Q\displaystyle\sum_{\tau=-P+1}^{P-1}\bm{x}_t^H\textbf{D}^H_{\tau,q,q'}\bm{x}_t\textbf{D}_{\tau,q,q'},\
{\textbf{E}}\triangleq\text{mat}(\hat{\bm{\Psi}}\textbf{1}).
\end{align*}
\end{lemma}
{
To reduce the computational burden, matrices $\bm{\Phi}_1$, $\bm{\Phi}_2$, and $\bm{\Phi}_3$ can be rewritten, respectively, as
\begin{equation*}
    \begin{aligned}
        \bm{\Phi}_1&= \textbf{I}_L \otimes \left([\textbf{a}_1\textbf{a}_1^H,\dots,\textbf{a}_U\textbf{a}^H_U](\textbf{u}_1\otimes \textbf{I}_N)\right) \\        
        \bm{\Phi}_2&= \displaystyle\sum_{q=1}^{Q}\displaystyle\sum_{\substack{q'=1, \\ q'\neq q}}^Q [\textbf{J}_{-P+1},\dots,\textbf{I}_L,\dots,\textbf{J}_{P-1}](\textbf{u}_2\otimes \textbf{I}_L) \otimes \textbf{a}_q\textbf{a}^H_{q'} \\        
        \bm{\Phi}_3&= \displaystyle\sum_{q=1}^{Q}[\textbf{J}_{-P+1},\dots,\dots,\textbf{J}_{P-1}](\textbf{u}_3\otimes \textbf{I}_L) \otimes \textbf{a}_q\textbf{a}^H_{q}       
    \end{aligned}
\end{equation*}
where \vspace{-1mm}
\begin{equation*}
    \begin{aligned}
        \textbf{u}_1 &= \begin{bmatrix}
            \text{vec}^H(\textbf{a}_1\textbf{a}_1^H) \\
            \vdots \\
            \text{vec}^H(\textbf{a}_U\textbf{a}_U^H)
        \end{bmatrix} \text{vec}(\textbf{X}_t\textbf{X}_t^H)\\
        \textbf{u}_2&=\text{vec}^H(\textbf{X}_t^H\textbf{a}_q\textbf{a}^H_{q'}\textbf{X}_t)\\
        &\cdot[\text{vec}(\textbf{J}_{-P+1}),\dots,\text{vec}(\textbf{I}_L),\dots,\text{vec}(\textbf{J}_{P-1})]\\\textbf{u}_3&=\text{vec}^H(\textbf{X}_t^H\textbf{a}_q\textbf{a}^H_{q}\textbf{X}_t) \\
        &\cdot[\text{vec}(\textbf{J}_{-P+1}),\dots,\text{vec}(\textbf{J}_{-1}),\text{vec}(\textbf{J}_{1}),\dots,\text{vec}(\textbf{J}_{P-1})].
    \end{aligned}
\end{equation*}
Note we temporarily denoted $\textbf{a}(\theta_u)=\textbf{a}_u$ and $\textbf{a}(\theta_q)=\textbf{a}_q$ for brevity.
}
This majorizer is still quadratic, which is challenging to optimize under the constant modulus constraint. 
Thus, we further majorize the obtained quadratic function to lower its order as follows.
\begin{lemma}\label{lemma:lin_maj}
    Let $\hat{\bm{\Phi}}$ be a matrix such that $\hat{{\Phi}}_{i,j}=|{{\Phi}}_{i,j}|$ for any $i,j$.
    The quadratic function on the right-hand side of \eqref{ineq:obj_quad} is majorized by 
    \begin{equation}
       \bm{x}^H\bm{\Phi}\bm{x} \leq \underbrace{\Re\{\bm{x}^H\textbf{d}\}}_{\bar{g}(\bm{x})}+const, \vspace{-1.5mm}
    \end{equation}
where $\textbf{d}\triangleq2({\bm{\Phi}}-\text{diag}(\hat{\bm{\Phi}}\textbf{1}))\bm{x}_t$.   
\end{lemma}
\begin{proof}
By applying Lemma \ref{lemma:diag} and Lemma \ref{lemma:quad_maj}, we have
    \begin{align*}
    \bm{x}^H{\bm{\Phi}}\bm{x} 
     &\leq \underbrace{\bm{x}^H\text{diag}(\hat{\bm{\Phi}}\textbf{1})\bm{x}}_{\textbf{1}^T\hat{\bm{\Phi}}\textbf{1}}+
     \Re\{\bm{x}^H\underbrace{2({\bm{\Phi}}-\text{diag}(\hat{\bm{\Phi}}\textbf{1}))\bm{x}_t}_{\textbf{d}
    }\}  \\ \nonumber
    &\quad+\bm{x}_t^H(\text{diag}(\hat{\bm{\Phi}}\textbf{1})-{\bm{\Phi}})\bm{x}_t
    = \Re\{\bm{x}^H\textbf{d}\}+const.
\end{align*}\vspace{-2mm}
\end{proof}
\begin{theorem}
Given the constant modulus constraint, the objective function can be majorized as
    \begin{equation}\label{eq:majorizer}
       \omega_{bp}\tilde{g}_{bp}(\bm{x})+\omega_{ac} g_{ac}(\bm{x})+\omega_{cc} g_{cc}(\bm{x}) \leq \bar{g}(\bm{x})+const,
    \end{equation}    
    where $ \bar{g}(\bm{x}) = \Re\{\bm{x}^H\textbf{d}\} =\Re\{\textbf{d}^H\bm{x}\}$.
\end{theorem}

\subsection{Solution via the Method of Lagrange Multipliers}

Now, using \eqref{eq:majorizer}, problem \eqref{eq:prob_formulation2} can be reformulated as
\begin{equation}\label{eq:MM_formulation}
\begin{aligned}
& \underset{\bm{x}}{\min}
& &  \Re\{\textbf{d}^H\bm{x}\} \\ 
& \text{s.t.}
& &  \Re\{\tilde{\textbf{h}}_{\ell,m}^H\bm{x}\}\geq \tilde{\Gamma}_m, \ \forall \ell,m\\
& & & | {x}_{n} | = 1, \ \forall n=1,2,\dots,LN_T \\
\end{aligned}
\end{equation}
The majorized objective can be rewritten as $\Re\{\textbf{d}^H\bm{x}\}=\sum_{\ell=1}^{L}\Re\{\textbf{d}_{\ell}^H\bm{x}_{\ell}\}$, where $\bm{x}_{\ell}$ and $\textbf{d}_{\ell}$ are the $\ell$th subvectors of $\bm{x}=[\bm{x}_1^H,\bm{x}_2^H,\dots,\bm{x}_L^H]^H$ and $\textbf{d}=[\textbf{d}_1^H,\textbf{d}_2^H,\dots,\textbf{d}_{L}^H]^H$, respectively.
Also, from \eqref{ineq:CI}, we have $\Re\{\tilde{\textbf{h}}_{\ell,m}^H\bm{x}\}=\Re\{\hat{\textbf{h}}_{\ell,m}^H\bm{x}_{\ell}\}$.
Hence, the problem \eqref{eq:MM_formulation} can be rewritten as
\begin{equation}\label{eq:MM_formulation2}
\begin{aligned}
& \underset{\{\bm{x}_{\ell}\}_{\ell=1}^L}{\min}
& &  \sum_{\ell=1}^{L}\Re\{\textbf{d}_{\ell}^H\bm{x}_{\ell}\} \\ 
& \text{s.t.}
& &  \Re\{\hat{\textbf{h}}_{\ell,m}^H\bm{x}_{\ell}\}\geq \tilde{\Gamma}_m, \ \forall \ell,m\\
& & & | {x}_{\ell,n} | = 1, \ \forall n=1,2,\dots,N_T \\
\end{aligned}
\end{equation}
where ${x}_{\ell,n}$ is the $n$th entry of $\bm{x}_{\ell}$.
Since $\bm{x}_1,\bm{x}_2,\dots,\bm{x}_L$ are independent of each other in \eqref{eq:MM_formulation2}, the problem \eqref{eq:MM_formulation2} can be split into $L$ independent subproblems as
\begin{equation}\label{eq:MM_subprob}
\begin{aligned}
& \underset{\bm{x}_{\ell}}{\min}
& &  \bar{g}_{\ell}(\bm{x}_{\ell})  \\ 
& \text{s.t.}
& &  {h}_{\ell,m}(\bm{x}_{\ell})\leq 0, \ \forall m=1,2,\dots,2K\\
& & & | {x}_{\ell,n} | = 1, \ \forall n=1,2,\dots,N_T \\
\end{aligned}
\end{equation}
where $\bar{g}_{\ell}(\bm{x}_{\ell})=\Re\{\textbf{d}_{\ell}^H\bm{x}_{\ell}\}$ and ${h}_{\ell,m}(\bm{x}_{\ell})=-\Re\{\hat{\textbf{h}}_{\ell,m}^H\bm{x}_{\ell}\}+\tilde{\Gamma}_m$.
The Lagrange dual problem for \eqref{eq:MM_subprob} is given by
\begin{equation}\label{eq:MM_dual}
\begin{aligned}
& \underset{\bm{\nu}_{\ell}}{\sup}\ \underset{\bm{x}_{\ell}}{\min}
& &  \bar{g}_{\ell}(\bm{x}_{\ell}) + \displaystyle\sum_{m=1}^{2K} \nu_{\ell,m} {h}_{\ell,m}(\bm{x}_{\ell})\\ 
& \text{s.t.}
& &  | {x}_{\ell,n} | = 1, \ \forall n=1,2,\dots,N_T \\
& & & \nu_{\ell,m} \geq 0 , \ \forall m,\ell 
\end{aligned}
\end{equation}
where ${x}_{\ell,n}$ is the $n$th entry of $\bm{x}_{\ell}$ and $\bm{\nu}_{\ell}=[\nu_{\ell,1},\nu_{\ell,2},\dots,\nu_{\ell,2K}]$ is the Lagrange multiplier vector with ${\nu}_{\ell,m}$ being the Lagrange multiplier for the communication constraint ${h}_{\ell,m}(\bm{x}_{\ell})\leq 0$.
The inner problem of \eqref{eq:MM_dual} has a linear objective with a constant modulus constraint.
Thus, the optimal solution to the inner problem can be expressed as $\bm{x}_{\ell}^{opt}(\bm{\nu}_{\ell})=\exp{\left(j 
\angle\left(\sum_{m=1}^{2K}\nu_{\ell,m}\hat{\textbf{h}}_{\ell,m}-\textbf{d}_{\ell}\right)\right)}$.


Strong duality between the primal and dual problems holds \cite{he2022qcqp} if there exists a solution $\bm{\nu}_{\ell}$ that satisfies the following conditions:
\begin{align}
    &\bm{x}_{\ell}(\bm{\nu}_{\ell})=\exp{\left(j 
    \angle\left(\displaystyle\sum_{m=1}^{2K}\nu_{\ell,m}\hat{\textbf{h}}_{\ell,m}-\textbf{d}_{\ell}\right)\right)}, 
    \label{eq:sol_to_x}\\
    &0\leq \nu_{\ell,m} \leq \infty,\ {h}_{\ell,m}(\bm{x}_\ell(\bm{\nu}_{\ell}))\leq 0, \forall m =1,2,\dots,2K \label{ineq:dual_feasible}\\
    &\nu_{\ell,m}{h}_{\ell,m}(\bm{x}_\ell(\bm{\nu}_{\ell}))=0, \forall m =1,2,\dots,2K \label{eq:KKT}.
\end{align}

A solution satisfying \eqref{eq:sol_to_x} and \eqref{eq:KKT} always exists, given $\nu_{\ell,m}<\infty$ for all $\ell,m$.
Assuming that the feasible set is strictly feasible, we have $\lim_{\nu_{\ell,m}\rightarrow\infty}{h}_{\ell,m}(\bm{x}_{\ell}(\bm{\nu_{\ell}}))={h}_{\ell,m}\left(\exp{\left(j \angle\hat{\textbf{h}}_{\ell,m}\right)}\right)<0$ for any $\ell,m$.
Hence, there exists finite $\bm{\nu}_{\ell}$ that satisfies equation \eqref{eq:KKT}, leading to strong duality.
Using this fact, we focus on solving the dual problem rather than directly solving the primal problem.
Given the closed-form solution to the inner problem \eqref{eq:sol_to_x}, the dual problem \eqref{eq:MM_dual} can be reduced to finding optimal Lagrange multipliers $\bm{\nu}_{\ell}$ that satisfy conditions \eqref{ineq:dual_feasible} and \eqref{eq:KKT}.
With this in mind, the dual problem can be reformulated as
\begin{equation}\label{eq:iMM_dual2}
\begin{aligned}
& \underset{\bm{\nu}_{\ell}}{\sup}\ 
& &\bar{g}_{\ell}(\bm{\nu}_{\ell}) + \displaystyle\sum_{m=1}^{2K} \nu_{\ell,m} {h}_{\ell,m}(\bm{\nu}_{\ell})\\ 
& \text{s.t.} & & \nu_{\ell,m} \geq0 ,\ {h}_{\ell,m}(\bm{\nu}_{\ell})\leq 0, \ \forall m=1,2,\dots,2K \\
& & &\nu_{\ell,m}{h}_{\ell,m}(\bm{\nu}_{\ell})=0, \forall m =1,2,\dots,2K
\end{aligned}
\end{equation}
For ease of notation, $\bar{g}_{\ell}(\bm{x}_{\ell}(\bm{\nu}_{\ell}))$ and ${h}_{\ell,m}(\bm{x}_{\ell}(\bm{\nu}_{\ell}))$ are denoted by $\bar{g}_{\ell}(\bm{\nu}_{\ell})$ and ${h}_{\ell,m}(\bm{\nu}_{\ell})$, respectively.

\setlength{\textfloatsep}{0pt}
\begin{algorithm}
\DontPrintSemicolon
\SetNlSty{textbf}{\{}{\}}
\SetAlgoNlRelativeSize{0}
\SetNlSty{}{}{}

\caption{
$2K$-Dimension 
Bisection Method for Finding Dual Variables \label{alg:K-Bisection}}

\textbf{Input:} Lagrange multiplier vector {$\bm{\nu}_{\ell}$, stopping thresholds $\epsilon_2$, $\epsilon_3$, $\varepsilon$}

\textbf{Initialization:} $i=0$; $\bm{\nu}_{\ell}[0] = \bm{\nu}_{\ell}$, $\hat{g}_{\ell}[0]=\infty$;
With slight abuse of notation, ${h}_{\ell,m}(\nu')$ denotes ${h}_{\ell,m}(\bm{\nu}_{\ell})|_{\nu_{\ell,m}=\nu'}$\;

\Repeat{$|\hat{g}_{\ell}[i]-\hat{g}_{\ell}[i-1]|/|\hat{g}_{\ell}[i-1]|<\epsilon_2$}{
    \SetKwBlock{DoParallel}{do in parallel}{end}
    \For{$m=1:2K$}{
        \lIf{${h}_{\ell,m}(0)\leq 0$}{$\nu_{\ell,m}=0$}
        \ElseIf{$\lim_{\nu_{\ell,m} \rightarrow \infty}|h_{\ell,m}(\nu_{\ell,m})| \leq \varepsilon$}{Stop Algorithm \ref{alg:K-Bisection}}
        \Else{
            $\nu^l=0$, $\nu^u=1$;\;
            \lIf{${h}_{\ell,m}(\nu^u)\leq 0$}{$\nu^u=1$}
            \Else{
                \lRepeat{${h}_{\ell,m}(\nu^u)\leq0$}{$\nu^u=2\nu^u$}
                $\nu^l=\nu^u/2$\;
            }
            \Repeat{$|{h}_{\ell,m}(\nu_{\ell,m})+\epsilon_3/2|<\epsilon_3/2$}{
                $\nu_{\ell,m}=(\nu^l+\nu^u)/2$;\;
                \lIf{${h}_{\ell,m}(\nu_{\ell,m})>0$}{$\nu^l=\nu_{\ell,m}$}
                \lElse{$\nu^u=\nu_{\ell,m}$}
            }
        }
    }
    Update $i \gets i+1$,\ set $\bm{\nu}_{\ell}[i]=[\nu_1,\dots,\nu_{2K}]$\;
    Update $\hat{g}_{\ell}[i]=\bar{g}_{\ell}(\bm{\nu}_{\ell}\;[i])+\sum_{m=1}^{2K}\nu_{\ell,m}{h}_{\ell,m}(\bm{\nu}_{\ell}[i])$\;
}

{\textbf{Output:} Recover a solution $\bm{x}$ from $\bm{\nu}[i]$ and \eqref{eq:sol_to_x}}

\end{algorithm}
\setlength{\textfloatsep}{0pt}
\begin{algorithm}[t]
\DontPrintSemicolon
\SetNlSty{textbf}{\{}{\}}
\SetAlgoNlRelativeSize{0}
\SetNlSty{}{}{}

\caption{Proposed MM Algorithm\label{alg:MM}}

{\textbf{Input:} Initial point $\bm{x}_0$, stopping threshold $\epsilon_4$}

\textbf{Initialize:} Set $t=0$, $\bm{x}^{(t)}=\bm{x}_0$, $g[t]=\infty$\;

\Repeat{$|g[t]-g[t-1]|/|g[t-1]|\leq \epsilon_4$}{
    $t \gets t+1$\;
    Update $\bm{x}_{1},\dots,\bm{x}_L$ using \eqref{eq:majorizer} and Algorithm \ref{alg:K-Bisection}\;
    $\bm{x}^{(t)}\gets [\bm{x}_1^H,\bm{x}_2^H,\dots,\bm{x}_L^H]^H$ \; 
    $g[t]\gets\omega_{bp}\tilde{g}_{bp}(\bm{x}^{(t)})+\omega_{ac} g_{ac}(\bm{x}^{(t)})+\omega_{cc} g_{cc}(\bm{x}^{(t)})$\;
}

{\textbf{Output:} $\textbf{X}=\text{mat}(\bm{x}^{(t)})$}

\end{algorithm}

The problem \eqref{eq:iMM_dual2} can be solved via a coordinate ascent method where one Lagrange multiplier is optimized at a time with the other Lagrange multipliers fixed.
For updating each coordinate, we use a modified version of the bisection algorithm in \cite{he2022qcqp}, as described in Algorithm \ref{alg:K-Bisection}.
Once the Lagrange multiplier $\bm{\nu}_{\ell}$ is obtained, $\bm{x}_{\ell}$ can be recovered using \eqref{eq:sol_to_x}.
{
Note that $\bm{x}_{1},\bm{x}_{2},\dots,\bm{x}_{L}$ can be updated in parallel to accelerate the algorithm. 
}
The solution $\bm{x}_t$ for the $t$-th MM iteration can be obtained by concatenating the subvectors as $\bm{x}_t=[\bm{x}_1^H,\bm{x}_2^H,\dots,\bm{x}_L^H]^H$.
This iterative process continues until the objective value converges.
The final converged solution can be reshaped into a matrix as $\textbf{X}=\text{mat}(\bm{x}_t)$.
The overall iterative solution is described in Algorithm \ref{alg:MM}.



\subsection{Complexity, Convergence and Parallelization}\label{sec:subsec:mm_complexity}
Now we analyze the complexity of our proposed MM algorithm.
The proposed MM algorithm comprises the majorization process and the bisection algorithm for solving the dual problem.
The majorization process involves computation of the matrices $\bm{\Psi}$, $\bm{\Phi}$, and the vector $\textbf{d}$.
The matrix $\bm{\Psi}$ can be precomputed since it is independent of variable $\bm{x}_t$.
Thus, we focus on analyzing the complexity of computing $\bm{\Phi}$ and $\textbf{d}$.
The computation of ${\bm{\Phi}}$ requires the computations of ${\bm{\Phi}}_1$, ${\bm{\Phi}}_2$, ${\bm{\Phi}}_3$, which cost $O(UL^2N_T^2)$, $O(Q(2P-1)L^2N_T^2)$ and $O(Q(Q-1)(2P-1)L^2N_T^2/2)$, respectively.
The computation of $\textbf{d}$ involves evaluating $2(\bm{\Phi}-\text{diag}(\hat{\bm{\Phi}}))\bm{x}_t$, which costs $O(L^2N_T^2)$.
Thus, the overall computational complexity of the majorization process can be expressed as $O(UL^2N_T^2+Q^2PL^2N_T^2)$.

Next, we analyze the complexity of the bisection algorithm.
The bisection algorithm requires the evaluation of ${h}_{\ell,m}(\bm{\nu}_{\ell})$, which costs $O(N_T^2)$.
The considered bisection method terminates when the constraint $h_{\ell,m}(\bm{\nu}_{\ell})$ sufficiently approaches zero.
This differs from the traditional bisection method that terminates when the length of the search interval falls below a threshold.
Thus, it is difficult to acquire an analytical bound on the worst-case iteration number due to the nonlinear relationship between $h_{\ell,m}(\bm{\nu}_{\ell})$ and the Lagrange multiplier.
However, the combination of MM and the considered bisection methods has empirically shown superior convergence rates to the penalty convex–concave procedure (CCP) method and semi-definite relaxation (SDR) \cite{he2022qcqp}.

Although the proposed MM algorithm provides monotone convergence and admits parallel dual updates, it may suffer from the quadratic complexity scaling when the array size or block length is large. 
In the next section, we develop a low-complexity alternative for scenarios with large MIMO dimensions.

\section{Fast LADMM Algorithm}\label{sec:LADMM}

While the MM solution in Section \ref{sec:MM} yields desirable convergence properties like monotonic descent, its computational cost scales with the MIMO dimension, making it less attractive for massive arrays or long CPIs.
To support such large-scale regimes, we develop a low-complexity alternative based on ADMM.
Prior MIMO radar work \cite{cheng2017constant} converted the nonconvex quartic objective to biconvex quadratic subproblems and iteratively found the critical points of the quadratic subproblems. 
While convenient, this approach relies on matrix inversions, which may become computationally prohibitive in high dimensions.
To avoid these expensive operations, we employ linearized ADMM (LADMM), which has been shown to be effective for large-scale nonconvex QCQPs \cite{konar2017fast}.
LADMM replaces quadratic terms with first-order approximations, enabling simple proximal updates without matrix inversions while preserving ADMM's decomposability.

\vspace{-1mm}

\subsection{ADMM Transformation} 
First, we reformulate the problem in \eqref{eq:prob_formulation2} by introducing auxiliary variables $\bm{u}\in \mathbb{C}^{LN_T}$, $\bm{v}\in \mathbb{C}^{LN_T}$, and ${z}_{\ell,m}\in\mathbb{C}$ for $\ell=1,2,\dots,L$ and $m=1,2,\dots,2K$ as \vspace{-2mm}
\begin{equation}\label{eq:ADMM_prob}
\begin{aligned}
 \underset{\bm{x},\bm{u},\bm{v},\{\bm{z}_{\ell}\}_{\ell=1}^{L}}{\min}   g(\bm{x},\bm{v}) \\
\quad \text{s.t.} \quad 
\Re\{ {z}_{\ell,m}\} & \geq\tilde{\Gamma}_{m}, \, \forall \ell,m\\
& | u_{n} | = 1, \ \forall n \\
 \bm{x} & =\bm{v}, \;\; \bm{u}  =\bm{v}, \\
 {z}_{\ell,m} & = \tilde{\textbf{h}}_{\ell,m}^H\bm{x}, \, \forall \ell,m
\end{aligned}
\end{equation}
where $\bm{z}_{\ell} = [z_{\ell,1},z_{\ell,2},\dots,z_{\ell,2K}]$, \vspace{-2mm}
\begin{equation}
    \begin{aligned}  
    g(\bm{x},\bm{v})&\triangleq\omega_{bp}\tilde{g}_{bp}(\bm{x},\bm{v})  +\omega_{ac}g_{ac}(\bm{x},\bm{v})+\omega_{cc}g_{cc}(\bm{x},\bm{v}), \\
    \tilde{g}_{bp}(\bm{x},\bm{v})
    &\triangleq\displaystyle\sum_{u=1}^U|\bm{x}^H\textbf{B}_u\bm{v}|^2, \\ 
    g_{ac}(\bm{x},\bm{v})&\triangleq\displaystyle\sum_{q=1}^{Q}\displaystyle\sum_{\substack{\tau=-P+1, \\ \tau\neq 0}}^{P-1}\left|\bm{x}^H\textbf{D}_{\tau,q,q}\bm{v}\right|^2,\text{ and }\\
    g_{cc}(\bm{x},\bm{v}) &\triangleq\displaystyle\sum_{q=1}^{Q-1}\displaystyle\sum_{\substack{q'=1, \\ q'\neq q}}^Q\displaystyle\sum_{\tau=-P+1}^{P-1}\left|\bm{x}^H\textbf{D}_{\tau,q,q'}\bm{v}\right|^2.    
\end{aligned}
\end{equation}
By substituting one $\bm{x}$ with an auxiliary variable $\bm{v}$, the objective becomes bi-convex, i.e, convex in $\bm{x}$ with $\bm{v}$ fixed and in $\bm{v}$ with $\bm{x}$ fixed \cite{cheng2017constant} and an unconstrained quadratic problem with respect to $\bm{x}$ or $\bm{v}$.
Moreover, the constant modulus and QoS constraints are decoupled through the introduced auxiliary variables $\bm{u}$ and $\{\bm{z}_{\ell}\}_{\ell=1}^{L}$.
The scaled augmented Lagrangian function for \eqref{eq:ADMM_prob} can be rewritten as \vspace{-1mm}
\begin{equation*}
    \begin{aligned}
         &\mathcal{L}(\bm{x},\bm{v},\bm{u},\bm{z},\bm{\rho},{\bm{\eta}}_1,{\bm{\eta}}_2)=  g(\bm{x},\bm{v}) 
         \\  \nonumber    &
    +\frac{\mu_1}{2}(\left\Vert\bm{x}-\bm{v}+{{\bm{\eta}}}_1\right\Vert^2-\left\Vert{{\bm{\eta}}}_1\right\Vert^2) 
    \\     &
    +\frac{\mu_2}{2}(\left\Vert\bm{u}-\bm{v}+{\bm{\eta}}_2\right\Vert^2-\left\Vert {\bm{\eta}}_2\right\Vert^2)  
    \\ \nonumber
    &+\frac{\mu_3}{2}\left(\left\Vert \bm{z}-\tilde{\textbf{H}}\bm{x}+\bm{\rho}\right\Vert^2-\left\Vert \bm{\rho}\right\Vert^2\right), \vspace{-3mm}
    \end{aligned}
\end{equation*}
where $\mu_1, \mu_2,\mu_3\in \mathbb{R}^+$ are the scalar penalty parameters, $\bm{\eta}_1,\bm{\eta}_2\in \mathbb{C}^{LN_T \times 1}$ are the Lagrange multipliers for the equality constraints $\bm{x}=\bm{v}$ and $\bm{u}=\bm{v}$, respectively, $\rho_{\ell,m}\in \mathbb{C}^{}$ is the Lagrange multiplier for the equality constraint ${z}_{\ell,m}=\tilde{\textbf{h}}_{\ell,m}^H\bm{x}$ for all $\ell,m$, $\bm{z}=[\bm{z}^H_1,\bm{z}^H_2,\dots,\bm{z}^H_{L}]^H$ is a concatenated vector of $\bm{z}_1,\bm{z}_2,\dots,\bm{z}_{L}$, $\bm{\rho}=[\bm{\rho}_1^H,\bm{\rho}^H_2,\dots,\bm{\rho}^H_{L}]^H$ is the Lagrange multiplier vector for the equality constraints with $\bm{\rho}_\ell=[\rho_{\ell,1},\rho_{\ell,2},\dots,\rho_{\ell,2K}]^T$, and $\tilde{\textbf{H}}=[\textbf{H}^H_{1},\textbf{H}^H_2,\dots,\textbf{H}^H_{L}]^H$ is a stacked matrix with $\textbf{H}_{\ell}=[\tilde{\textbf{h}}_{\ell,1},\tilde{\textbf{h}}_{\ell,2},\dots,\tilde{\textbf{h}}_{\ell,2K}]^H$.
Accordingly, the problem \eqref{eq:ADMM_prob} can be decomposed into multiple subproblems and written in iterative form as
\begin{align}
 \bm{x}^{(i+1)}&:=\underset{\bm{x}}{\arg\min}  \ \mathcal{L}\left(\bm{x},[\bm{v},\bm{u},\bm{z},\bm{\rho},{\bm{\eta}}_1,{\bm{\eta}}_2]^{(i
)}\right), \label{eq:ADMM_x}\\ 
\bm{v}^{(i+1)}&:=\underset{\bm{v}}{\arg\min} \
 \mathcal{L}\left(\bm{x}^{(i+1)},\bm{v},[\bm{u},\bm{z},\bm{\rho},{\bm{\eta}}_1,{\bm{\eta}}_2]^{(i
)}\right), \label{eq:ADMM_v}\\
\bm{u}^{(i+1)}&:=\underset{\bm{u}\in\mathcal{R}_u}{\arg\min}
 \mathcal{L}\left([\bm{x},\bm{v}]^{(i+1)},\bm{u},[\bm{z},\bm{\rho},{\bm{\eta}}_1,{\bm{\eta}}_2]^{(i
)}\right), \label{eq:ADMM_u}\\
z_{\ell,m}^{(i+1)}&:=\underset{z_{\ell,m}\in\mathcal{R}_m}{\arg\min}
 \mathcal{L}
 \left([\bm{x},\bm{v},\bm{u}]^{(i+1)},\bm{z},[\bm{\rho},{\bm{\eta}}_1,{\bm{\eta}}_2]^{(i
)}\right), \label{eq:ADMM_z}
\\
{\bm{\eta}}_1^{(i+1)}&:={\bm{\eta}}_1^{(i)}+\bm{x}^{(i+1)}-\bm{v}^{(i+1)}, \\
{\bm{\eta}}_2^{(i+1)}&:={\bm{\eta}}_2^{(i)}+\bm{u}^{(i+1)}-\bm{v}^{(i+1)},
\\
\rho_{\ell,m}^{(i+1)}&:=\rho_{\ell,m}^{(i)}+z_{\ell,m}^{(i+1)}-\tilde{\textbf{h}}_{\ell,m}^H\bm{x}^{(i+1)},
\end{align}
where $i$ is the LADMM iteration index, $\mathcal{R}_u=\{\bm{u}:|\bm{u}_n|=1, \ \forall n=1,2,\dots,LN_T\}$, and $\mathcal{R}_m=\{{z}:\Re\{{z}\}\geq \tilde{\Gamma}_m\}$.

\subsection{Linearized Proximal Update of $\bm{x}^{(i+1)}$}
The subproblem \eqref{eq:ADMM_x} is an unconstrained quadratic optimization, which admits a closed-form solution via matrix inversion.
However, the complexity of matrix inversion increases cubically with the variable size $LN_T$, which can be unaffordable when the variable size is large.
To circumvent the computational burden, we employ a linearized update.
Instead of finding the critical point of $\mathcal{L}$ w.r.t. $\bm{x}$, we approximate the quadratic penalty term using a first-order Taylor expansion at the current point $\bm{x}^{(i)}$ and add a proximal term $\frac{\mu_x}{2}\Vert\bm{x}-\bm{x}^{(i)} \Vert^2$.
This allows subproblem \eqref{eq:ADMM_x} to be rewritten as
\begin{equation}\label{eq:prox_x}
    \bm{x}^{(i+1)}=\arg \min_{\bm{x}} \left\langle \nabla_{\bm{x}}\mathcal{L}^{(i)},\bm{x}-\bm{x}^{(i)} \right\rangle  +\frac{\mu_x}{2}\Vert\bm{x}-\bm{x}^{(i)} \Vert^2
\end{equation}
where $\mu_x$ is the proximal step size for $\bm{x}$ and $\nabla_{\bm{x}}\mathcal{L}^{(i)}$ is the gradient w.r.t. $\bm{x}$ at iteration $i$, which is given by
    \begin{equation}
    \begin{aligned}
        \nabla_{\bm{x}}\mathcal{L}^{(i)}\triangleq& \nabla_{\bm{x}}g\left(\bm{x}^{(i)},\bm{v}^{(i)}\right)+\mu_1\left(\bm{x}^{(i)}-\bm{v}^{(i)}+\bm{\eta}_1^{(i)}\right)\\
        &+\mu_3\tilde{\textbf{H}}^H\left(\tilde{\textbf{H}}\bm{x}^{(i)}-\bm{z}^{(i)}+\bm{\rho}^{(i)}\right).
    \end{aligned}    
\end{equation}
The computation of $\nabla_{\bm{x}}g\left(\bm{x}^{(i)},\bm{v}^{(i)}\right)$ is provided in Appendix \ref{sec:appendix:gradient}.
The approximated problem \eqref{eq:prox_x} yields a closed-form solution:
\begin{equation}
    \bm{x}^{(i+1)}=\bm{x}^{(i)}-\frac{1}{\mu_x}\nabla_{\bm{x}}\mathcal{L}^{(i)}.
\end{equation}


\subsection{Linearized Proximal Update of $\bm{v}^{(i+1)}$}
Similar to the subproblem \eqref{eq:ADMM_x}, the subproblem \eqref{eq:ADMM_v} for $\bm{v}$ is an unconstrained quadratic problem, which can be updated via a linearized update as
\begin{equation}
    \bm{v}^{(i+1)}=\bm{v}^{(i)}-\frac{1}{\mu_v}\nabla_{\bm{v}}\mathcal{L}^{(i)}
\end{equation}
where $\mu_v$ is the proximal step size for $\bm{v}$ and $\nabla_{\bm{v}}\mathcal{L}^{(i)}$ is the gradient w.r.t. $\bm{v}$ at iteration $i$, which is given by
    \begin{equation}
    \begin{aligned}
        \nabla_{\bm{v}}\mathcal{L}^{(i)}\triangleq& \nabla_{\bm{v}}g\left(\bm{x}^{(i+1)},\bm{v}^{(i)}\right)+\mu_2\left(\bm{v}^{(i)}-\bm{u}^{(i)}-\bm{\eta}_2^{(i)}\right).
    \end{aligned}    
\end{equation}
The computation of $\nabla_{\bm{v}}g\left(\bm{x}^{(i+1)},\bm{v}^{(i)}\right)$ is provided in Appendix \ref{sec:appendix:gradient}.

   
\subsection{Update of ${z}_{\ell,m}^{(i+1)}$}
Next, ignoring the irrelevant variables, the subproblem \eqref{eq:ADMM_z} for the auxiliary variable $z_{\ell,m}$ can be rewritten as 
\begin{equation}\label{eq:ADMM_sub_prob3}
\underset{{z}_{\ell,m}\in\mathcal{R}_m}{\min}
\ \left|{z}_{\ell,m}-\tilde{\textbf{h}}_{\ell,m}^H\bm{x}^{(i+1)}+{{\rho}}^{(i)}_{\ell,m}\right|^2.
\end{equation}
The above subproblem is convex due to the convex objective and constraint.
Thus, the closed-form solution can be readily obtained from the Karush-Kuhn-Tucker (KKT) condition as
\begin{equation}\label{eq:ADMM_sub_prob3_sol}
    {z}_{\ell,m}=\begin{cases}
         \tilde{\textbf{h}}_{\ell,m}^H\bm{x}^{(i+1)}-{\rho}^{(i)}_{\ell,m},\  \text{if }  \Re\{\tilde{\textbf{h}}_{\ell,m}^H\bm{x}^{(i+1)}-\rho^{(i)}_{\ell,m}\} \geq \tilde{\Gamma}_{m}\\      
         {\begin{array}{c}        {\tilde{\textbf{h}}_{\ell,m}^H\bm{x}^{(i+1)}-{\rho}^{(i)}_{\ell,m}+\tilde{\Gamma}_{m}} \\
        {-\Re\{\tilde{\textbf{h}}_{\ell,m}^H\bm{x}^{(i+1)}-{\rho}^{(i)}_{\ell,m}\}}
        \end{array},}
    \text{ otherwise}. 
    \end{cases}
\end{equation}

\subsection{Update of $\bm{u}^{(i+1)}$}

The subproblem \eqref{eq:ADMM_u} for the auxiliary variable $\bm{u}$ can be simplified as
\begin{equation}\label{eq:ADMM_sub_prob4}
\begin{aligned}
& \underset{{\bm{u}}}{\min}
& &  \left\Vert \bm{u}-\bm{v}^{(i+1)}+{\bm{\eta}}_2^{(i)}\right\Vert  \\
& \text{s.t.}
& & |u_n|=1, \ \forall n=1,2,\dots,LN_T.
\end{aligned}
\end{equation}
The solution to \eqref{eq:ADMM_sub_prob4} is given by
$\bm{u}^{(i+1)}=e^{j\angle({\bm{v}^{(i+1)}-{\bm{\eta}}_2^{(i)})}}$, which was proven in \cite{he2022qcqp}.
{
}

The subproblems can be iteratively solved until the stopping criterion is satisfied.
Then, we can recover the converged solution by reshaping the vector $\bm{x}$ into the matrix $\textbf{X}$, as described in Algorithm \ref{alg:LADMM}.


\setlength{\textfloatsep}{0pt}
\begin{algorithm}
\caption{Proposed LADMM Algorithm}
\label{alg:LADMM}
{\textbf{Input:} Initial point $\bm{x}_0$, stopping threshold $\epsilon_1$ 
}

\textbf{Initialize:}
$i \gets 0$, $g[i]=\infty$, 
$\bm{x}^{(i)}=\bm{x}_0$,
 ${\bm{\eta}}_1^{(i)}={\bm{\eta}}_2^{(i)}=\textbf{0}_{LN_T\times 1}$,  $\rho_{\ell,m}^{(i)} = 0$, ${z}_{\ell,m}^{(i)} = \Re\{\tilde{\textbf{h}}_{\ell,m}^H\bm{x}^{(i)}\}$, $\bm{v}^{(i)} = \bm{x}^{(i)}$, $\bm{u}^{(i)} = \bm{v}^{(i)}$ \\

\Repeat{$|g[i]-g[i-1]|/|g[i-1]|\leq \epsilon_1$}{
update $\bm{x}^{(i+1)}$, $\bm{v}^{(i+1)}$ and ${z}_{\ell,m}^{(i+1)}$
\\
update $\bm{u}^{(i+1)} \gets e^{j\angle({\bm{v}^{(i+1)}-{\bm{\eta}}_2^{(i)})}}$\\
update ${\bm{\eta}}_1^{(i+1)}\gets{\bm{\eta}}_1^{(i)}+\bm{x}^{(i+1)}-\bm{v}^{(i+1)}$\\
update ${\bm{\eta}}_2^{(i+1)}\gets{\bm{\eta}}_2^{(i)}+\bm{u}^{(i+1)}-\bm{v}^{(i+1)}$\\
update ${{\rho}}_{\ell,m}^{(i+1)}\gets {{\rho}}_{\ell,m}^{(i)}+{z}_{\ell,m}^{(i+1)}-\tilde{\textbf{h}}^H_{\ell,m}\bm{x}^{(i+1)}$\\
set $i \gets i+1$ \\
$g[i]\gets\omega_{bp}\tilde{g}_{bp}(\bm{x}^{(i)})+\omega_{ac} g_{ac}(\bm{x}^{(i)})+\omega_{cc} g_{cc}(\bm{x}^{(i)})$\
}
{\textbf{Output:} $\textbf{X}=\text{mat}(\bm{x}^{(i)})$}

\end{algorithm}


\subsection{Complexity Analysis}\label{sec:subsec:admm_complexity}


Each LADMM iteration requires updating variables $\bm{x}, \bm{v},\bm{u},\bm{z},\bm{\rho},\bm{\eta}_1,\bm{\eta}_2$.
Unlike standard ADMM, the proposed LADMM algorithm avoids the computationally expensive matrix inversion $\mathcal{O}(N_T^3 L^3)$ through linearized proximal updates.
Consequently, the computational complexity is dominated by the gradient calculation $\nabla_{\bm{x}}\mathcal{L}_{\rho}$ in the $\bm{x}$ and $\bm{v}$ update steps.
The gradient computation for the beam pattern cost function involves matrix multiplications between the steering matrix and the waveform matrix, with a complexity of $\mathcal{O}(U N_T L)$.
For the ISL cost functions, we project the waveform onto the steering vectors and then utilize the fast Fourier transform (FFT) to accelerate the convolution operations.
Calculating the auto-correlation gradients requires $\mathcal{O}(QN_TL+QL \log L)$ operations, while the cross-correlation gradients, which involve interactions between all target pairs, scale as $\mathcal{O}(Q(Q-1)L (\log L + N_T))$.
The gradient contributions from the communication constraint involve matrix-vector products with the channel matrix $\tilde{\textbf{H}}$, costing $\mathcal{O}(K N_T L)$.
The updates for the auxiliary variables $\bm{z}$ and $\bm{u}$ involve simple element-wise projections with linear complexities of $\mathcal{O}(KL)$ and $\mathcal{O}(N_T L)$, respectively.
Combining these results, the total computational complexity of each LADMM iteration is given by $ \mathcal{O}(U N_T L + Q^2 L (\log L + N_T))$.
This represents a significant reduction compared to the cubic and quadratic scaling of the standard ADMM and the MM solution, respectively, making the proposed algorithm scalable for larger variable sizes $LN_T$.

\section{Simulation Results}
\label{sec:sims}
In this section, we evaluate the proposed algorithms through simulations.
We use the following setting unless otherwise specified.
The waveform contains 32 subpulses, i.e., $L=32$, and the largest range bin of interest is $P=8$ \cite{wang2012design}.
Also, the transmit power is $P_T=1$ 
and the noise variance for the communication users $\sigma^2=0.01$ \cite{liuJointTransmitBeamforming2020a}.
The transmit array is equipped with $N_T=8$ antennas with half-wavelength spacing \cite{wang2012design}.
We consider the uncorrelated Rayleigh channel for the communication channel of each user.
{
\color{black}
We use 500 channel realizations to evaluate the average performance of the proposed algorithms unless otherwise specified.
}
We set the discretized angle range to be $[0^{\circ},180^{\circ}]$ with the angle resolution of $0.5^{\circ}$, i.e., $\theta_u=(u/2)^{\circ}$ for $u=1,2,\dots,360$. 
For the reference beam pattern, we consider a rectangular beam pattern, which is given by 
\begin{equation}
    G_d(\theta)=
    \begin{cases}
    1,& \text{if } \theta_q - \Delta_{\theta}/2\leq\theta\leq  \theta_q + \Delta_{\theta}/2 \ \forall q,\\
    0,              & \text{otherwise},
\end{cases}
\end{equation}
where $\Delta_{\theta}$ is the beam width.
We consider two target directions, i.e, $Q=2$ each at angles $\theta_1=-30^{\circ}$ and $\theta_2=40^{\circ}$.
The beam width $\Delta_{\theta}$ is set to $20^{\circ}$.
The termination thresholds are set to $\epsilon_1=10^{-4}$, $\epsilon_2=\epsilon_3 = 10^{-4}$, and $\epsilon_4=3 \times 10^{-6}$.
We configure the penalty parameters for the LADMM algorithm as $\mu_1=\mu_2=\mu_3=10^4$.

For baselines, we use a radar-only scheme that solves \eqref{eq:prob_formulation2} without the communication constraints, to verify the radar-communication trade-off.
Also, we compare the proposed algorithm to the algorithm in \cite{liuDualFunctionalRadarCommunicationWaveform2021}, which optimizes the beam pattern shaping cost on a symbol-by-symbol basis without suppressing correlations between symbols under a per-user CI constraint.

{
\color{black}
Initialization significantly impacts the convergence speed of the proposed algorithms.
Thus, we solve the following problem to find an initial point for the proposed algorithms: 
\begin{equation}
\begin{aligned} 
& \underset{\bm{x},\varphi}{\max}
& &  \varphi \\ 
& \text{s.t.}
& &  \Re\{\tilde{\textbf{h}}_{\ell,m}^H\bm{x}\}\geq \varphi, \ \forall \ell,m\\
& & & | {x}_{n} | \leq 1, \ \forall n=1,2,\dots,LN_T.
\end{aligned}
\end{equation}
The above problem is convex, which can be solved using numerical tools like CVX.
}

\begin{figure}
     \centering
     \begin{subfigure}[b]{0.24\textwidth}
         \centering
            \center{\includegraphics[width=.95\linewidth]{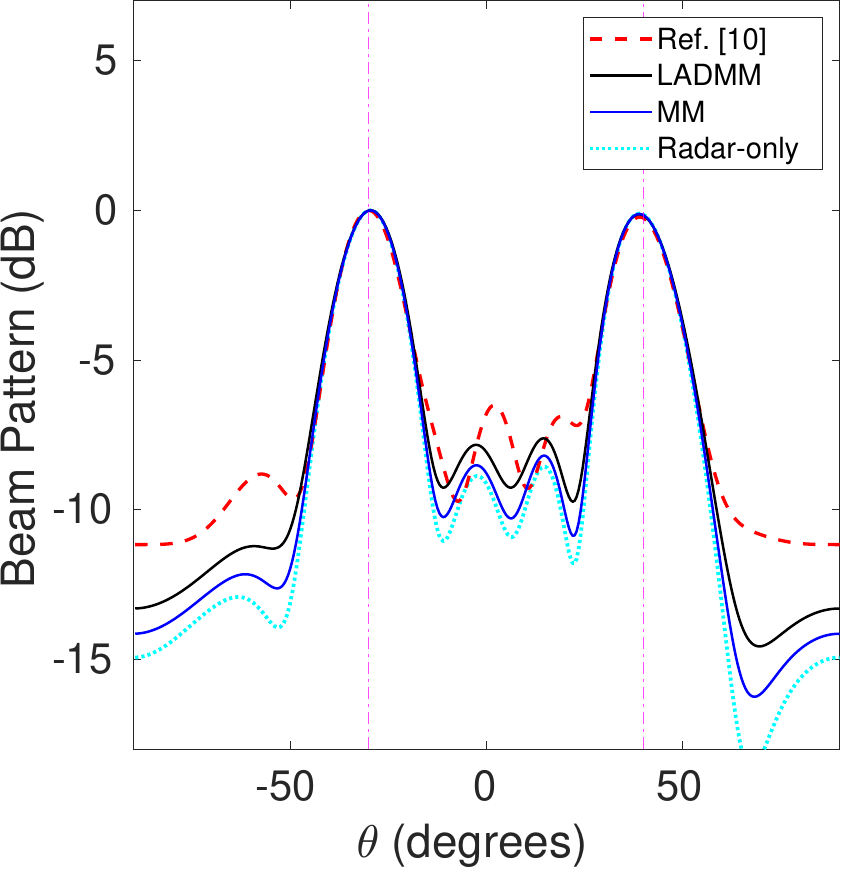}}
            \caption{ $K=2, \gamma_k=\SI{6}{\decibel}$.}
            \label{fig:Beam_Pattern1}
         \end{subfigure}
     \hfill
     \begin{subfigure}[b]{0.24\textwidth}
         \centering
            \center{\includegraphics[width=.95\linewidth]{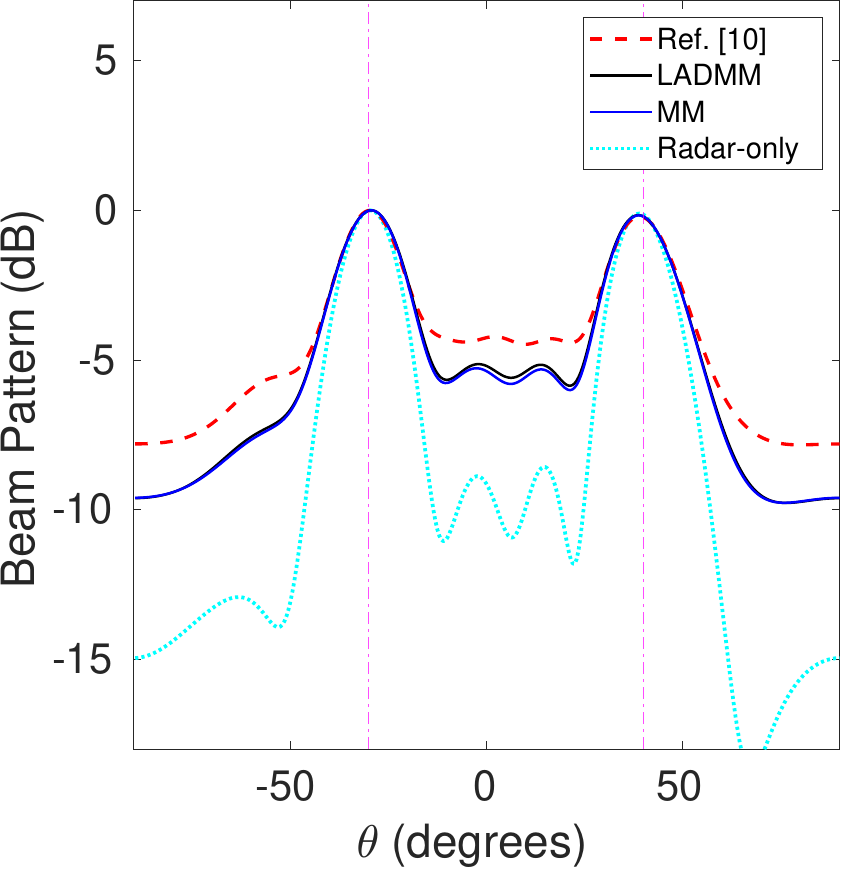}}
            \caption{ $K=4, \gamma_k=\SI{12}{\decibel}$.}
            \label{fig:Beam_Pattern2}
         \end{subfigure}
     \caption{\small  Synthesized beam patterns for two communication parameter sets.
     }
     \label{fig:Beam_Pattern}
\end{figure}



Figs. \ref{fig:Beam_Pattern1} and \ref{fig:Beam_Pattern2} compare the beam patterns designed by the proposed algorithms, the per-symbol design \cite{liuDualFunctionalRadarCommunicationWaveform2021}, and the radar-only scheme, for $K=2, \gamma_k=\SI{6}{\decibel}$ and $K=4, \gamma_k=\SI{12}{\decibel}$. 
The weights for the cost functions are $(\omega_{bp},\omega_{ac},\omega_{cc})=(1,4,4)$.
For both communication configurations, the radar-only scheme outperforms DFRC schemes in beam pattern approximation because it has no communication constraints.
When $K=2,\gamma_k=\SI{6}{\decibel}$, the beam patterns of the proposed methods approach that of the radar-only scheme, while the per-symbol design baseline suffers from relatively higher sidelobe levels.
The baseline \cite{liuDualFunctionalRadarCommunicationWaveform2021} focuses on the symbol-by-symbol beam pattern shaping, which can be seen as a myopic approach.    
In contrast, our approach optimizes the average beam pattern for the entire block, resulting in lower spatial sidelobes.
When $K=4,\gamma_k=\SI{12}{\decibel}$, we observe a similar trend where the proposed approach maintains lower sidelobes than the baseline \cite{liuDualFunctionalRadarCommunicationWaveform2021}. 
The overall sidelobes levels increased compared to the previous figure, except for the radar-only scheme.
This suggests that the difficulty of beam pattern shaping increases as communication requirements become more demanding. 
For both cases, the MM solution outperforms the LADMM solution in terms of beam pattern approximation.

\begin{figure}
     \centering
     \begin{subfigure}[b]{0.24\textwidth}
     \centering
            {\includegraphics[width=.95\linewidth]{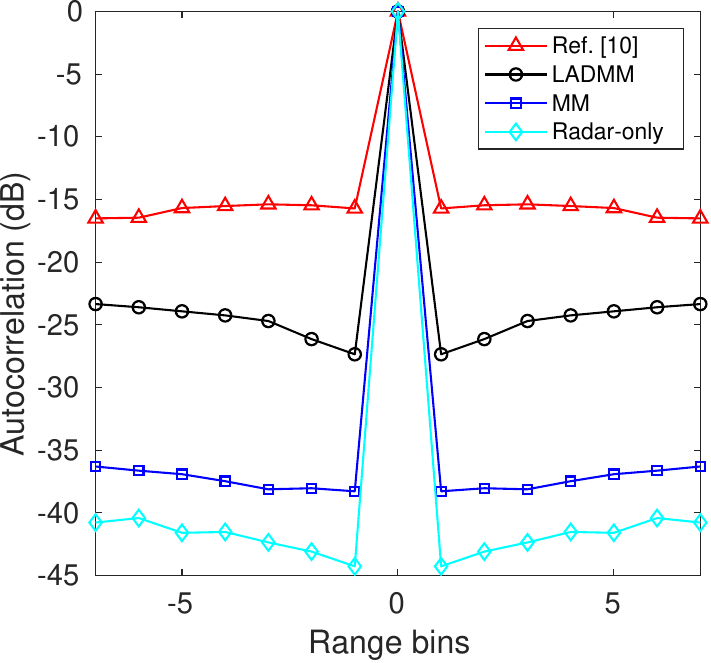}}
            \caption{ 
            $K=2$, $\gamma_k=\SI{6}{\decibel}$.
            }
            \label{fig:AC1_Set1}
     \end{subfigure}
     \hfill
     \begin{subfigure}[b]{0.24\textwidth}
            \center{\includegraphics[width=.95\linewidth]{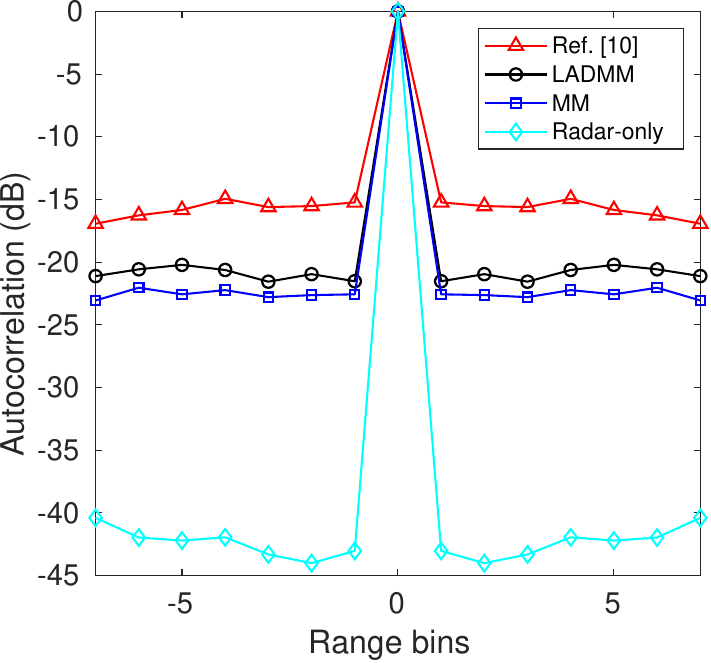}}
            \caption{ 
            $K=4$, $\gamma_k=\SI{12}{\decibel}$.
            }
            \label{fig:AC1_Set2}
     \end{subfigure}
     \begin{subfigure}[b]{0.24\textwidth}
     \centering
             {\includegraphics[width=.95\linewidth]{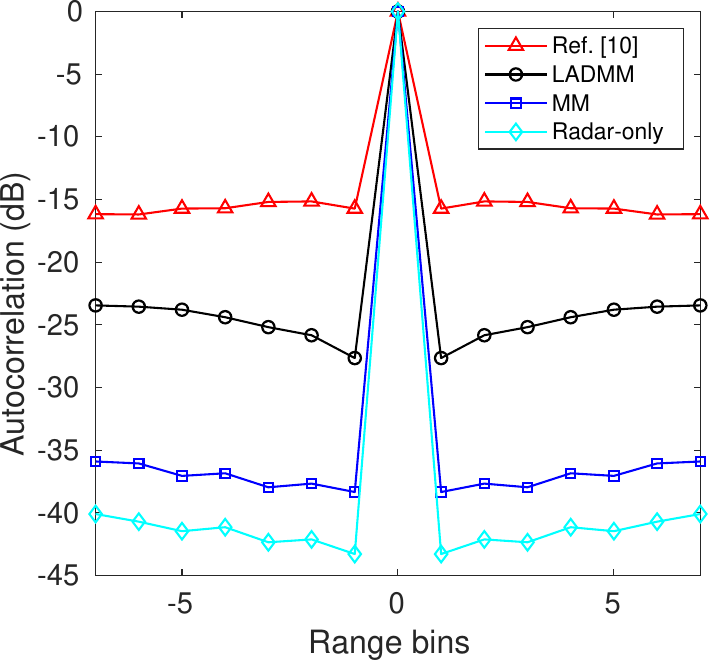}}
             \caption{ 
             $K=2$, $\gamma_k=\SI{6}{\decibel}$.
             }\label{fig:AC2_Set1}
     \end{subfigure}     
     \hfill
     \begin{subfigure}[b]{0.24\textwidth}
             \center{\includegraphics[width=.95\linewidth]{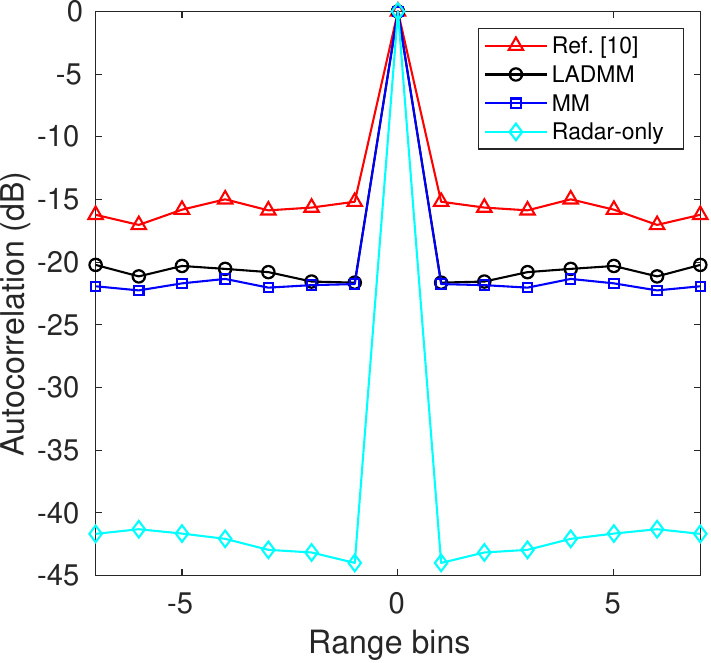}}
             \caption{
             $K=4$, $\gamma_k=\SI{12}{\decibel}$.
             } \label{fig:AC2_Set2}
     \end{subfigure}
     
     \caption{\small 
     Autocorrelation at target angles (\ref{fig:AC1_Set1})(\ref{fig:AC1_Set2}) $\theta_1=-30^{\circ}$ 
     and (\ref{fig:AC2_Set1})(\ref{fig:AC2_Set2}) $\theta_2=40^{\circ}$
     for two communication parameter sets. 
     }
     \label{fig:auto}
\end{figure}

\begin{figure}
     \centering
          \begin{subfigure}[b]{0.24\textwidth}
          \centering
             {\includegraphics[width=.95\linewidth]{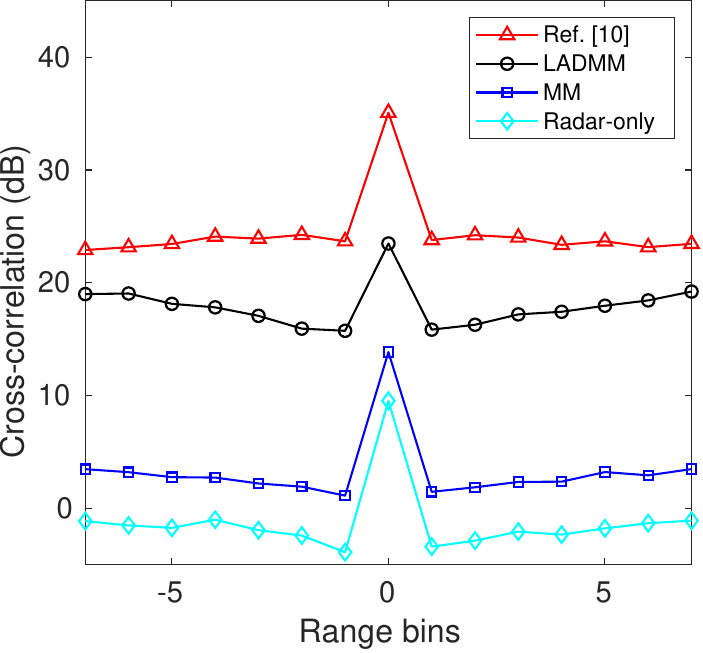}}
             \caption{
             $K=2$, $\gamma_k=\SI{6}{\decibel}$.
             } \label{fig:CC_Set1}
     \end{subfigure}
      \begin{subfigure}[b]{0.24\textwidth}
             \center{\includegraphics[width=.95\linewidth]{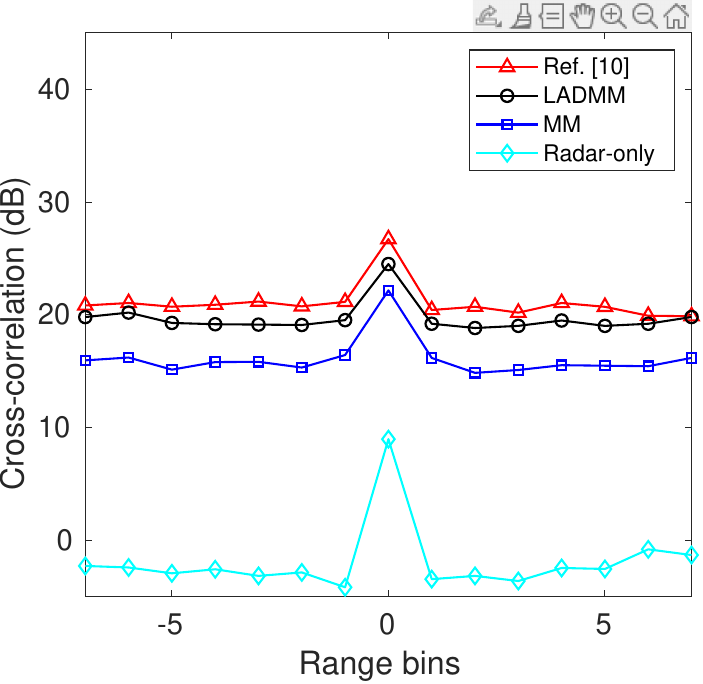}}
             \caption{
             $K=4$, $\gamma_k=\SI{12}{\decibel}$
             } \label{fig:CC_Set2}
     \end{subfigure}
     
     \caption{\small Cross-correlation between $\theta_1=-30^{\circ}$ and $\theta_2=40^{\circ}$ for two communication parameter sets.
     }
     \label{fig:cross}
\end{figure}

\begin{figure}
     \centering
     \begin{subfigure}[b]{0.24\textwidth}
         \centering
            \center{\includegraphics[width=.99\linewidth]
            {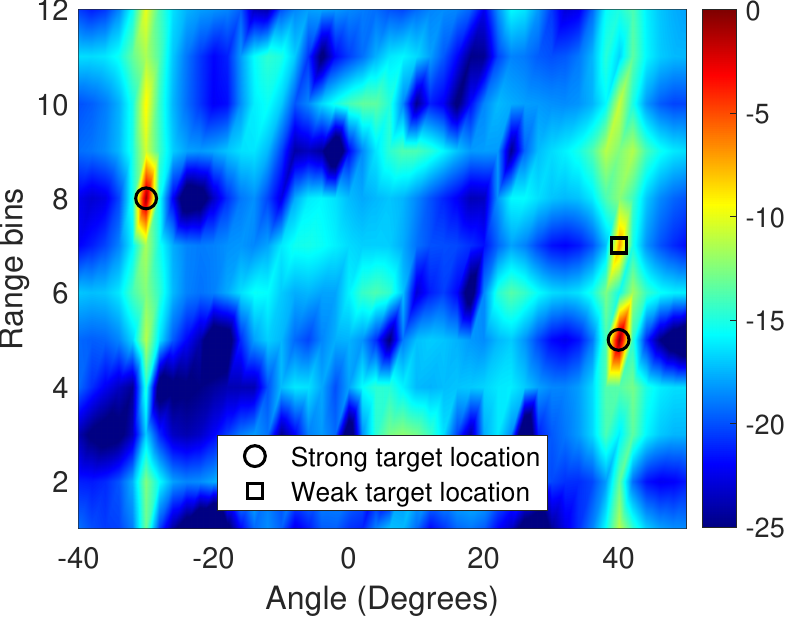}}
            \caption{\small LADMM (Proposed).}
            \label{fig:capon_BLP}
         \end{subfigure}
     \begin{subfigure}[b]{0.24\textwidth}
         \centering
            \center{\includegraphics[width=.99\linewidth]
            {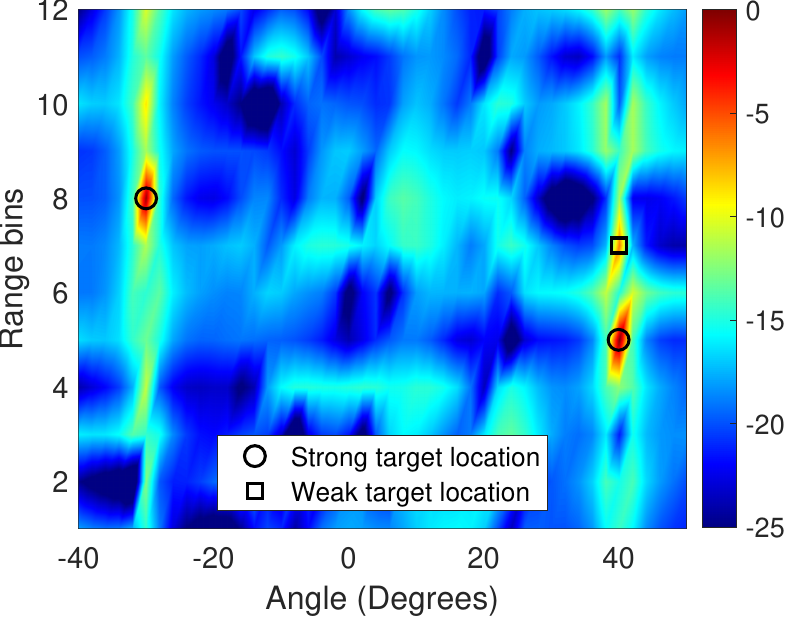}}
            \caption{\small MM (Proposed).}
            \label{fig:capon_BLP_MM}
         \end{subfigure}
     \begin{subfigure}[b]{0.24\textwidth}
         \centering
            \center{\includegraphics[width=.99\linewidth]
            {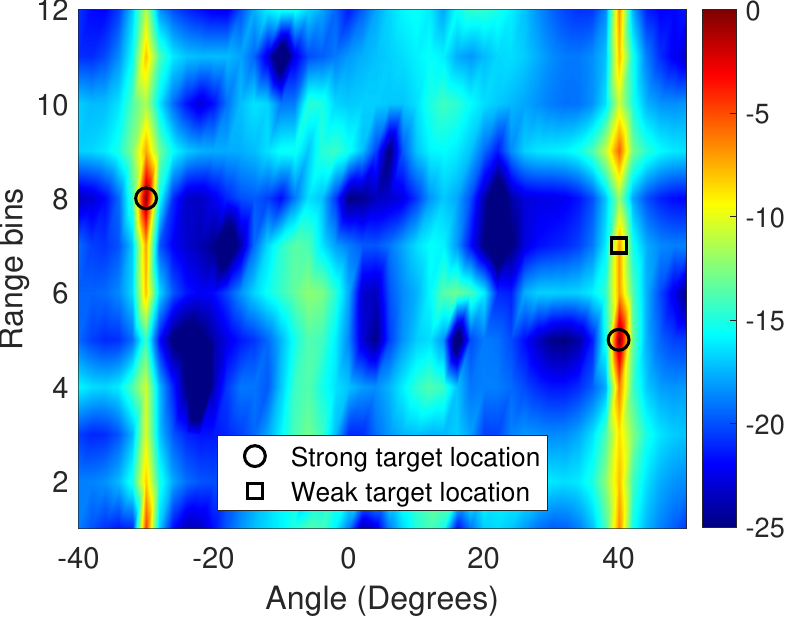}}
            \caption{ \small 
            Per-symbol design \cite{liuDualFunctionalRadarCommunicationWaveform2021}.
            }
            \label{fig:capon_SLP}
         \end{subfigure}
         \begin{subfigure}[b]{0.24\textwidth}
         \centering
            \center{\includegraphics[width=.99\linewidth]
            {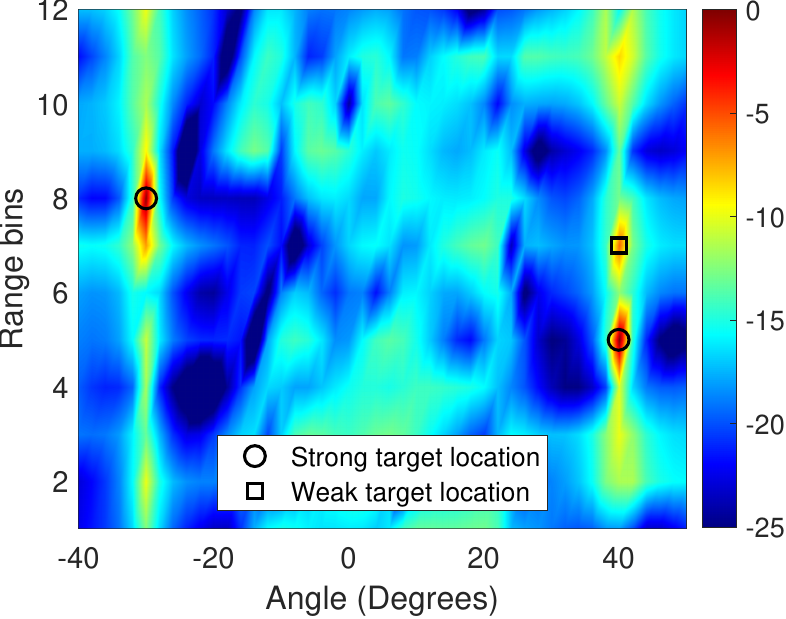}}
            \caption{ \small 
            Radar only.
            }
            \label{fig:capon_radar}
         \end{subfigure}
     \caption{\small Capon spectral images of the proposed waveforms and baselines in the angle and range domain for $K=2$ and $\Gamma=\SI{6}{\decibel}$.
     A weak target is placed at $(\theta_1=40^{\circ},\tau_1=7)$ and two strong targets are placed at $(\theta_2=-30^{\circ},\tau_2=8)$ and $(\theta_3=40^{\circ},\tau_3=5)$.
     }
     \label{fig:capon}
\end{figure}


{
Next, we evaluate the waveform correlation properties using the same setup described for Fig. \ref{fig:Beam_Pattern}.
Figs. \ref{fig:auto} and \ref{fig:cross} plot the autocorrelation and cross-correlation performance of the proposed method and baselines.
In all cases, the radar-only scheme outperforms the DFRC schemes in autocorrelation and cross-correlation, for the same reason as Fig. \ref{fig:Beam_Pattern}.
The per-symbol beam pattern design \cite{liuDualFunctionalRadarCommunicationWaveform2021} demonstrates the highest autocorrelation/cross-correlation sidelobe levels since it does not address waveform correlations.
In contrast, the proposed approach effectively reduces sidelobes owing to block-level ISL minimization.
It is important to note that the MM algorithm nearly matches the sidelobe suppression performance of the radar-only scheme when $K=2,\gamma_k=\SI{6}{\decibel}$, yielding a roughly $\SI{20}{\decibel}$ sidelobe reduction compared to the per-symbol design.
When $K=4,\gamma_k=\SI{12}{\decibel}$, the overall sidelobe levels of our approach become higher.
This implies suppressing sidelobes becomes harder as the communication requirements become tighter, accounting for the radar-communication trade-off.
Despite this, the proposed approach outperforms the baseline \cite{liuDualFunctionalRadarCommunicationWaveform2021} in terms of correlation for any configuration. 
Additionally, the MM solution achieves slightly lower sidelobe levels than the LADMM solution, consistent with the earlier results.
}

{
\color{black}
We perform a Capon spectral analysis \cite{xuRadarImagingAdaptive2006} to assess the positioning performance of the proposed waveforms.
The weights for the cost functions are set to $(\omega_{bp},\omega_{ac},\omega_{cc})=(1,10,10)$.
For each angle-range pair, we averaged the Capon estimates over 1000 noise realizations per channel realization.
We configured a weak target at $(\theta_1,\tau_1)=(40^{\circ},7)$ and two strong targets at $(\theta_2,\tau_2)=(-30^{\circ},8)$, $(\theta_3,\tau_3)=(40^{\circ},5)$.
We set the RCS of the strong target to be $\SI{6}{\decibel}$ higher than that of the weak target.

Figs. \ref{fig:capon_BLP} and \ref{fig:capon_SLP} illustrate the Capon estimates at different angle and range bins, generated by the proposed algorithms, respectively, for $K=2$ and $\Gamma=\SI{6}{\decibel}$.
All values are normalized to the maximum Capon amplitude and then converted to the dB scale.
{
From all results, two strong peaks appear at the strong target locations.
By contrast, the weak target shows a weaker response in all four images.
The per-symbol design result contains two strong peaks but with a broader spread compared to the radar-only and proposed schemes due to higher cross- and autocorrelations.
As a result, the weak target is masked by the strong targets' sidelobes, blocking its detection. 
By contrast, the proposed waveforms eliminate any false peaks and reduce dispersion around the targets due to their sharper mainlobes and reduced sidelobes.
}


}

\begin{figure}[t]
 \centering
    \center{\includegraphics[width=.7\linewidth]{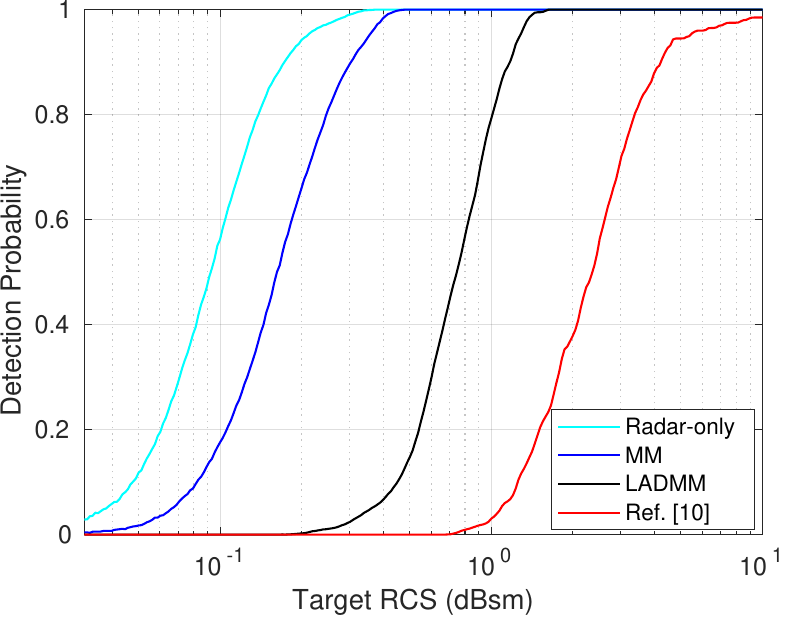}}
    \caption{\small 
    Detection probability of the target at $(\theta_1,\tau_1)=(40^{\circ},7)$ with varying target RCS values. The clutter setting remains the same as Fig. \ref{fig:capon}.
    }
    \label{fig:Pd}
 \end{figure}
{\color{black}

{
We now evaluate the target detection performance of the proposed waveforms for $K=2$ and $\gamma_k=\SI{6}{\decibel}$.
We apply a 1D cell-averaging constant false alarm rate (CA-CFAR) detector with a desired false alarm rate $P_{fa}=10^{-2}$.
Fig. \ref{fig:Pd} shows the detection probability for the target located at $(\theta_1,\tau_1)=(40^{\circ},7)$ with varying RCS values.
The MM algorithm closely approaches the radar-only performance and begin detecting the target at a very low RCS level.
By contrast, the per-symbol design fails to detect the target until the target's RCS becomes about $\SI{0}{\decibel sm}$, indicating that the weak target is masked by returns from adjacent
strong targets when correlation sidelobes are not explicitly controlled.
These results show that explicit correlation shaping in the proposed designs substantially improves weak-target sensitivity in multi-target scenarios.

}


}

\color{black}

{\color{black}
}

{

\begin{figure}[t]
 \centering
    \center{\includegraphics[width=.66\linewidth]{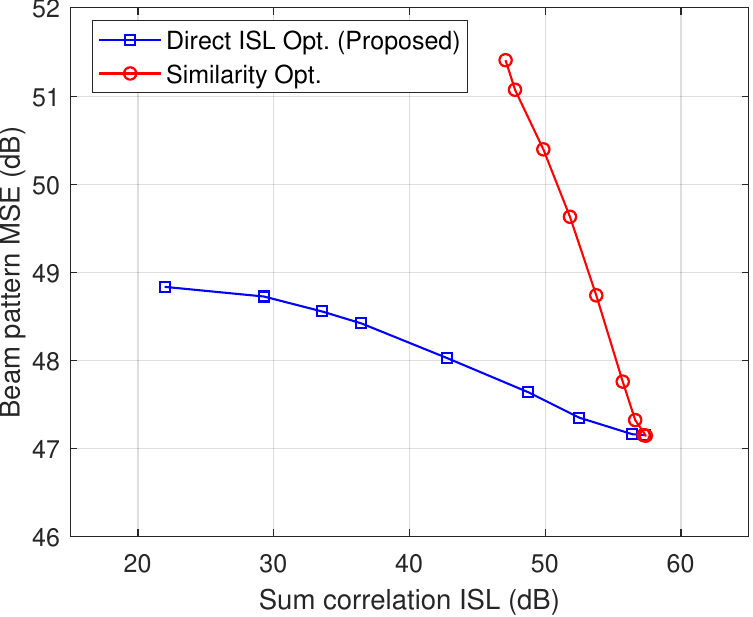}}
    \caption{\small Beam pattern MSE vs sum correlation ISL trade-off for the direct (ours) and indirect ambiguity shaping for $K=4$ and $\gamma_k=\SI{12}{\decibel}$.}
    \label{fig:trade_off}
    \vspace{-1mm}
 \end{figure}
We compare the trade-off between the beam pattern MSE and ISL cost functions for the direct (ours) and indirect (similarity-based) ambiguity function shaping approaches.
For the indirect approach, we consider the joint optimization of the beam pattern and similarity metrics.
For the indirect correlation shaping benchmark, we adopt the angular similarity metric, which is given by \cite{wen2023transmit}
$g_{sim}(\bm{x}) =\sum_{q=1}^Q \Vert \textbf{X}^H\textbf{a}(\theta_q)-\bm{x}_{ref}\Vert^2=\sum_{q=1}^Q \Vert \textbf{a}(\theta_q)\bm{x}-\bm{x}_{ref}\Vert^2$
where $\bm{x}_{ref}$ is a reference waveform (e.g., chirp).
The angular similarity metric is designed to quantify the similarity between the waveform radiated at the target angle and the reference sequence $\bm{x}_{ref}$.
We adopted a linear frequency-modulated (LFM) sequence as the reference waveform.
The $\ell$th entry of $\bm{x}_{ref}$ is given by  $x_{ref,\ell} =e^{j\pi (\ell-1)^2/L}$.
We minimize the weighted sum of the beam pattern cost function $\tilde{g}_{bp}(\bm{x})$ and the angular similarity metric $g_{sim}(\bm{x})$ using the proposed MM algorithm.
We set $K=4$ and $\gamma_k=\SI{12}{\decibel}$.


\begin{figure}[!t]
\centering
\begin{subfigure}[b]{0.158\textwidth}         
            \center{\includegraphics[width=.99\linewidth]{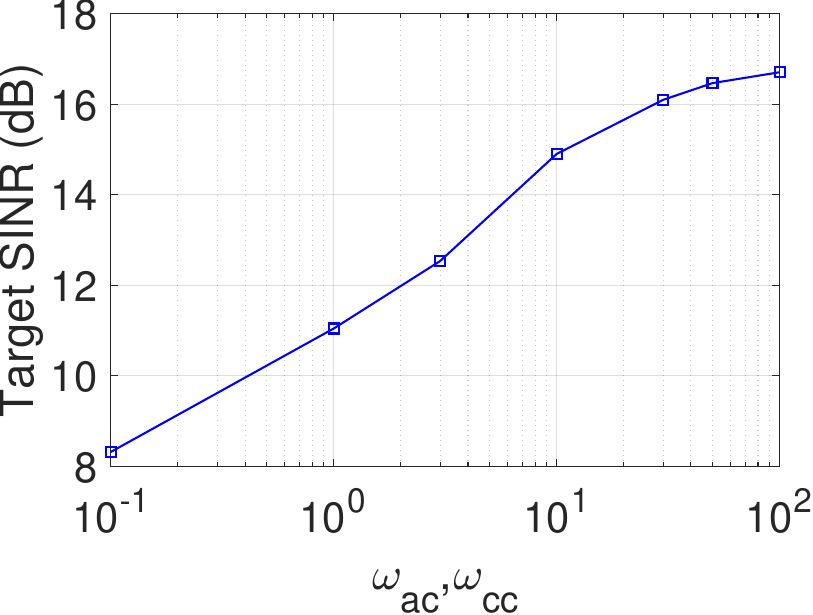}}
            \caption{ \small Scenario 1}
            \label{fig:clutter}
     \end{subfigure}
     \begin{subfigure}[b]{0.158\textwidth}         
            \center{\includegraphics[width=.99\linewidth]{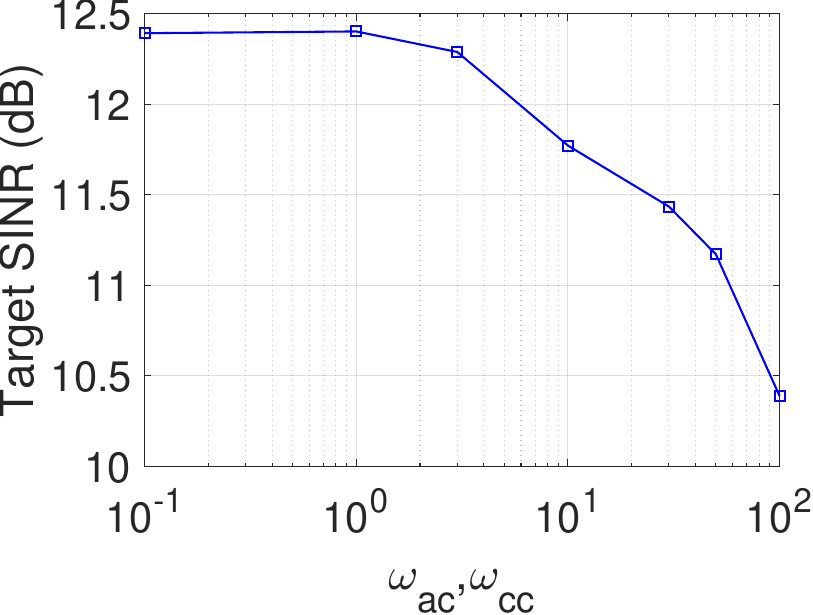}}
            \caption{ \small Scenario 2}
            \label{fig:jammer}
     \end{subfigure}
     \begin{subfigure}[b]{0.158\textwidth}         
            \center{\includegraphics[width=.99\linewidth]{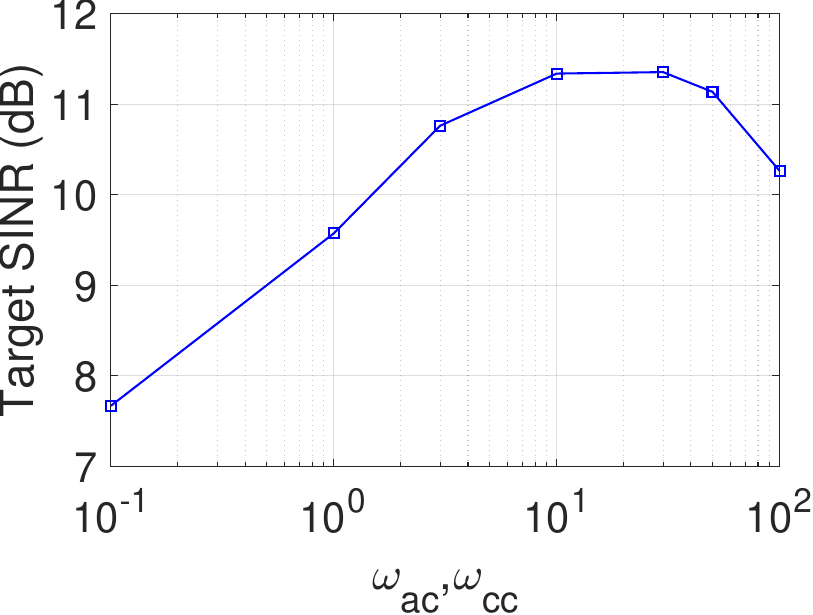}}
            \caption{ \small Scenario 3}
            \label{fig:jammer}
     \end{subfigure}
\caption{\small  Target SINR vs correlation ISL weights ($\omega_{ac},\omega_{cc}$) in three different scenarios.
The target/clutter settings are summarized in Table \ref{table:weight_strategy}.
}
\label{fig:weight}
\end{figure}

\begin{table}[!t]
  \centering
  \caption{Target and clutter setting in Fig \ref{fig:weight}.}
  \label{table:weight_strategy}
  \begin{tabular}{lcc}
    \toprule
    & \textbf{Target} & \textbf{Clutter/Jammer} \\
    \midrule
    Scenario 1 & $(\theta_1,\tau_1,\kappa_1)=(40^{\circ},7,1)$& \begin{tabular}{c}
$(\theta_2,\tau_2,\kappa_2)=(40^{\circ},5,4)$\\
 $(\theta_3,\tau_3,\kappa_3)=(-30^{\circ},8,4)$
\end{tabular}
\\ \midrule
Scenario 2  & $(\theta_1,\tau_1,\kappa_1)=(40^{\circ},7,1)$ & \begin{tabular}{c}
$(\theta_2,\tau_2,\kappa_2)=(0^{\circ},6,100)$\\
$(\theta_3,\tau_3,\kappa_3)=(-30^{\circ},8,1)$
\end{tabular}   \\ \midrule
    Scenario 3  & $(\theta_1,\tau_1,\kappa_1)=(40^{\circ},7,1)$ & \begin{tabular}{c}
    $(\theta_2,\tau_2,\kappa_2)=(40^{\circ},5,4)$\\
 $(\theta_3,\tau_3,\kappa_3)=(-30^{\circ},8,4)$\\
$(\theta_4,\tau_4,\kappa_4)=(0^{\circ},6,100)$\\
\end{tabular}   \\
    \bottomrule
  \end{tabular}
\end{table}


We now evaluate the impact of the weight parameters by fixing $\omega_{bp}=1$ and sweeping $\omega_{ac}$ and $\omega_{cc}$.
We configure the target parameter $(\theta_1,\tau_1,\kappa_1)=(40^{\circ},7,1)$ and vary the interference geometry across three cases.
Scenario 1 includes two clutter objects in each mainlobe at $\theta_2=-30^{\circ}$ and $\theta_3=40^{\circ}$ with range offsets.  Scenario 2 has only one clutter at $\theta_2=-30^{\circ}$ but adds a strong jammer in the sidelobe region at $\theta_3=0^{\circ}$.
Scenario 3 combines the two clutter objects in the mainlobes and the strong jammer at angle $0^{\circ}$.
The target/clutter settings are summarized in Table \ref{table:weight_strategy}.

Fig. \ref{fig:weight} plots the target SINR for $K=4$ and $\gamma_k=\SI{12}{\decibel}$ under the three interference scenarios.
In Scenario 1, SINR rises as the ISL weight increases from $\SI{8}{\decibel}$ to $\SI{16}{\decibel}$.
This is because the suppressed correlations reduce the interference power.
In Scenario 2, target SINR declines as correlation weights grow.
This is attributed to the increased spatial sidelobe level with the weights, which permits higher interference from the jammer.
Scenario 3 shows a mid-range plateau due to the combined clutter and jammer.
This implies that the balance of the spatial beam pattern and correlations is important in such scenario.

}

Fig. \ref{fig:convergence} compares the convergence properties of two majorizing functions based on the proposed diagonal matrix $\text{diag}(\hat{\textbf{Q}}\textbf{1})$ and the identity matrix $\lambda_{\textbf{Q}}\textbf{I}$ multiplied by the largest eigenvalue, when $K=2$ and $\gamma_k=\SI{6}{\decibel}$.
The proposed majorizer significantly increases the speed of convergence when compared to the largest eigenvalue-based majorizer.
This suggests that the proposed diagonal matrix offers a much tighter gap compared to the maximum eigenvalue reference. 
Moreover, consistent with the theory in \cite{he2022qcqp}, the MM algorithm shows a monotonic decrease in the objective value.

Table \ref{tab:comparison} compares the convergence times of the proposed algorithms and standard ADMM for variable sizes, i.e., $LN_T$.
The standard ADMM baseline uses the same biconvex formulation, but it seeks the critical points using matrix inversions to solve the quadratic subproblems.
We fix the transmit array size to $N_T=8$ and increase the block length $L$ from $L=4$ to $L=128$.
As expected, the convergence time tends to grow with the variable size $LN_T$. 
The LADMM algorithm achieves noticeably faster convergence than the MM algorithm due to low-complexity proximal iterations.
In particular, for $LN_T=1024$, the LADMM algorithm converges more than 300 times faster than the MM algorithm, demonstrating its computational complexity benefit.
Furthermore, the proposed LADMM algorithm converges significantly faster than the standard ADMM baseline, especially at large variable sizes.
This improvement comes from replacing the inversion-based subproblem solutions in ADMM with linearized proximal updates, avoiding the cubic complexity scaling of matrix inversions.

 \begin{figure}
 \centering
    \center{\includegraphics[width=.65\linewidth]{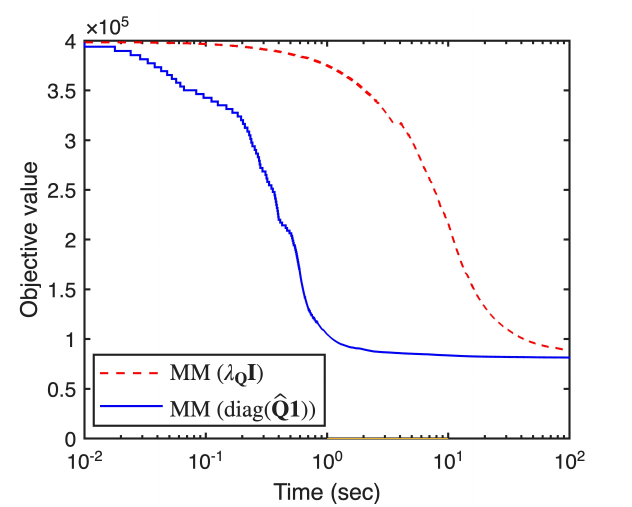}}
    \caption{\small Convergence of the proposed algorithms for $K=2$ and $\gamma_k=\SI{6}{\decibel}$. $\lambda_{\textbf{Q}}\textbf{I}$ and $\text{diag}(\hat{\textbf{Q}}\textbf{1})$ denote the results of the largest eigenvalue-based and proposed majorizers, respectively.}
    \label{fig:convergence}
 \end{figure}



\begin{table}[!t]
        \centering
        \renewcommand{\arraystretch}{1.4}
        \begin{tabular}{|l|c|c|c|c|c|c|}
            \hline
            $LN_T$ & $32$ & $64 $ & $128$ & $256$ & $512$ & $1024$ \\ \hline \hline          
            MM (Ours)& 0.573 & 3.952 & 8.54& 15.32 & 76.09 & 374.59 \\ \hline
            Standard ADMM & 0.124 & 0.73 & 1.83& 8.951& 84.4 & 764.71\\ \hline
            LADMM (Ours)& 0.077 & 0.179 & 0.32& 0.394& 0.83 &1.108\\ \hline             
        \end{tabular}
        \caption{CPU time (sec) comparison of the proposed algorithms and standard ADMM for different variable sizes $LN_T$.}
        \label{tab:comparison}
\end{table}

{

}
{
\section{Conclusion}
This paper investigated the design of constant modulus waveforms for DFRC systems.
We jointly optimized the spatial beam pattern and space-time correlations of the waveform to improve its space-time resolution.
For communications, we employed CI precoding that can expand the feasible region of the waveform variable, thereby enhancing the radar-communication trade-off.
To solve the formulated problem, we developed two algorithms, based on MM  and LADMM techniques, which offer reliable convergence and lower complexity, respectively.
Simulation results showed that the proposed waveforms outperform per-symbol optimized waveforms in terms of detection and imaging resolution, significantly enhancing sensing performance.

}

\appendices

\section{Space-Time Correlation Function}\label{sec:appendix_B}

The vector-form space-time correlation function can be derived using the basic properties of the trace and vectorization operators as \vspace{-2.7mm}
\begin{align*}
\chi_{\tau,q,q'}&=|\textbf{a}^H(\theta_q)\textbf{X}\textbf{J}_{\tau}\textbf{X}^H\textbf{a}(\theta_{q'})|^2 \\
&=|\text{Tr}\left(\textbf{X}^H\textbf{a}(\theta_{q'})\textbf{a}^H(\theta_q)\textbf{X}\textbf{J}_{\tau}\right)|^2 \\
&=|\text{Tr}\left((\textbf{a}(\theta_{q})\textbf{a}^H(\theta_{q'})\textbf{X})^H\textbf{X}\textbf{J}_{\tau}\right)|^2 \\
&=|\text{vec}^H\left(\textbf{a}(\theta_{q})\textbf{a}^H(\theta_{q'})\textbf{X}\right)\text{vec}\left(\textbf{X}\textbf{J}_{\tau}\right)|^2 \\
&=|\left((\textbf{I}_L \otimes \textbf{a}(\theta_{q})\textbf{a}^H(\theta_{q'}))\bm{x}\right)^H(\textbf{J}_{\tau}^T \otimes \textbf{I}_{N_T})\bm{x}|^2 \\
&= |\bm{x}^H \textbf{D}_{\tau,q,q'}\bm{x}|^2. \vspace{-2mm}
\end{align*}

\section{Proof of Theorem \ref{theorem:nonconvex}} \label{sec:appendix_C}

To show the nonconvexity of the feasible set, we transform the constraints in \eqref{eq:prob_formulation2} into a real-valued constraint as \vspace{-2mm}
\begin{equation}\label{eq:prob_formulation_real}
\begin{aligned}\vspace{-4mm}
 \bar{\textbf{h}}^T_{\ell,m}\bar{\bm{x}}_{\ell} \geq \tilde{\Gamma}_m \ \forall \ell,m, ~\bar{x}_{\ell,n}^2+\bar{x}_{\ell,n+LN_T}^2 = 1 \ \forall \ell,n,\vspace{-2mm}
\end{aligned}
\end{equation}
where $\bar{\bm{x}}_{\ell}=[\Re\{\bm{x}^T_{\ell}\},\Im\{\bm{x}_{\ell}^T\}]^T$, $\bar{\textbf{h}}_{\ell,m}=[\Re\{\tilde{\textbf{h}}^T_{\ell,m}\},\Im\{\tilde{\textbf{h}}^T_{\ell,m}\}]^T$, and $\bar{x}_{\ell,n}$ is the $n$th entry of $\bar{\bm{x}}_{\ell}$.
The feasible region of the constant modulus constraint takes the shape of a unit circle in the $n$th and $(n+LN_T)$th coordinates.
Moreover, the intersection of the linear communication constraints forms a polygon in the same coordinates.
Consequently, the intersection of the feasible sets turns out to be an arc of each circle. Thus, the feasible set is nonconvex, which proves the problem \eqref{eq:prob_formulation2} is nonconvex.
and NP-hard.

\section{Gradient Computation for LADMM}\label{sec:appendix:gradient}

The gradients for the beam pattern cost with respect to $\bm{x}$ and $\bm{v}$ are, respectively, given by \vspace{-2mm}
\begin{equation}\label{eq:grad_bp_x}
\begin{aligned}
    \nabla_{\bm{x}} g_{bp}(\bm{x}^{(i)},\bm{v}^{(i)})
&= 2\sum_{u=1}^{U}
c(\bm{x}^{(i)},\bm{v}^{(i)})\textbf{A}(\theta_u)\bm{v}^{(i)} \\
\nabla_{\bm{v}} g_{bp}(\bm{x}^{(i+1)},\bm{v}^{(i)})
&= 2\sum_{u=1}^{U}c(\bm{x}^{(i+1)},\bm{v}^{(i)})
\textbf{A}(\theta_u)\bm{x}^{(i)}
\end{aligned}
\end{equation}
where $c_u(\bm{x},\bm{v})=\Big(\bm{x}^H\textbf{A}(\theta_u)\bm{v}-\alpha^* G_{d,u}\Big)^{*}$ and $\alpha^*$ is the optimal scaling factor.
Note we treat $\alpha^*$ as a constant via the envelope theorem.
Let $\textbf{M}\triangleq[\textbf{a}(\theta_1),\dots,\textbf{a}(\theta_U)]\in\mathbb{C}^{N_T\times U}$ and define $\textbf{G}^{(i)} \triangleq \textbf{M}^H\mathbf{X}^{(i)}\in\mathbb{C}^{U\times L}$ and $\textbf{T}^{(i)} \triangleq \textbf{M}^H\mathbf{V}^{(i)}\in\mathbb{C}^{U\times L}$ where $\textbf{V}=\text{mat}(\bm{v})$.
Then $\bm{x}^H \textbf{A}_u\bm{v}=\sum_{\ell=1}^{L}{[\textbf{G}]^*_{u,\ell}}[\textbf{T}]_{u,\ell}$ and the matrix-form gradient can be written as
\begin{equation}
\begin{aligned}
    \nabla_{\mathbf{X}} g_{bp}(\bm{x}^{(i)},\bm{v}^{(i)})
&= 2\,\textbf{M}\Big(\mathbf{f}^{(i)}\mathbf{1}_L^T\odot \textbf{T}^{(i)}\Big) \\
\nabla_{\mathbf{V}} g_{bp}(\bm{x}^{(i+1)},\bm{v}^{(i)})
&= 2\,\textbf{M}\Big(\mathbf{q}^{(i)}\mathbf{1}_L^T\odot \textbf{G}^{(i)}\Big),
\end{aligned}
\end{equation}
where $\mathbf{f}^{(i)}\in\mathbb{C}^{U}$ has entries $f_u \triangleq c_u(\bm{x}^{(i)},\bm{v}^{(i)})$ and $\mathbf{q}^{(i)}\in\mathbb{C}^{U}$ has entries $q_u \triangleq c_u(\bm{x}^{(i+1)},\bm{v}^{(i)})$.
Finally, the matrix-form gradients can be converted to vector-form gradients as $\nabla_{\bm{x}} g_{bp}(\bm{x}^{(i)},\bm{v}^{(i)})=\mathrm{vec}\big(\nabla_{\mathbf{X}} g_{bp}(\bm{x}^{(i)},\bm{v}^{(i)})\big)$ and $\nabla_{\bm{v}} g_{bp}(\bm{x}^{(i)},\bm{v}^{(i)})=\mathrm{vec}\big(\nabla_{\mathbf{V}} g_{bp}(\bm{x}^{(i)},\bm{v}^{(i)})\big)$.


The gradient of the autocorrelation ISL with respect to $\bm{x}$ can be computed using the chain rule as
\begin{equation}
\begin{aligned}
    &\nabla_{\bm{x}} g_{ac}(\bm{x}^{(i)},\bm{v}^{(i)})
    =2 \sum_{q=1}^Q \sum_{\tau \in \mathcal{P}\setminus \{0\}} 
    \textbf{D}_{\tau,q,q} \bm{v}^{(i)} \underbrace{(\bm{v}^{(i)})^H\textbf{D}^H_{\tau,q,q}\bm{x}^{(i)}}
    \\
    &= 2 \sum_{q=1}^Q \text{vec}\left( \textbf{a}(\theta_q)  \sum_{\tau\in \mathcal{P}\setminus\{0\}} r^*_{\tau,q,q}(\bm{x}^{(i)},\bm{v}^{(i)}) \tilde{\bm{v}}_q^{(i)}\textbf{J}_{\tau}  \right)\\
     &= 2 \sum_{q=1}^Q \text{vec}\left( \textbf{a}(\theta_q)\left(\tilde{\textbf{r}}^{ac}_{q}(\bm{x}^{(i)},\bm{v}^{(i)})^*\circledast \tilde{\bm{v}}_q^{(i)}\right)  \right).
\end{aligned}
\end{equation}
where 
$r_{\tau,q,q}(\bm{x},\bm{v})=\bm{x}^H\textbf{D}_{\tau,q,q} \bm{v}$,
$\tilde{\bm{v}}_q^{(i)}=\textbf{a}^H(\theta_q)\textbf{V}^{(i)}$ is the beam-domain sequence at angle $\theta_q$ and $\mathcal{P}=\{-P+1,\dots,P-1\}$ is the correlation suppression window.
The masked autocorrelation vector $\tilde{\textbf{r}}^{ac}_{q}(\bm{x},\bm{v}) \in \mathbb{C}^{2L-1}$ is defined as $\tilde r^{ac}_{q}[\tau]= r_{\tau,q,q}(\bm x,\bm v)\, \mathbbm{1}_{\{\tau\in\mathcal P\setminus\{0\}\}}.
$
Similarly, the gradient for the cross-correlation ISL can be expressed as
\begin{equation}
\begin{aligned}
    &\nabla_{\bm{x}} g_{cc}(\bm{x}^{(i)},\bm{v}^{(i)}) \\
    &= 2  \sum_{q=1}^Q \sum_{\substack{q'=1\\ q'\neq q}}^Q\text{vec}\left( \textbf{a}(\theta_q)\left(\tilde{\textbf{r}}_{q,q'}^{cc}(\bm{x}^{(i)},\bm{v}^{(i)})^* \circledast \tilde{\bm{v}}_{q'}^{(i)}\right)  \right)
\end{aligned}
\end{equation}
where the vector $\tilde{\textbf{r}}^{cc}_{q,q'}(\bm{x},\bm{v})^* \in \mathbb{C}^{2L-1}$ is defined as $\tilde{{r}}^{cc}_{q,q'}[\tau]=r_{\tau,q,q}(\bm x,\bm v)\,\mathbbm{1}_{\{\tau\in\mathcal P\}}.$ 
The linear convolution operations can be accelerated via FFT.
Similarly, the gradients for the ISLs w.r.t. $\bm{v}$ are given by
\begin{equation}
    \begin{aligned}
    &\nabla_{\bm{v}} g_{ac}(\bm{x}^{(i+1)},\bm{v}^{(i)}) \\
    &= 2 \sum_{q=1}^Q \text{vec}\left( \textbf{a}(\theta_q)(\tilde{\textbf{r}}^{ac}_{q}\left(\bm{x}^{(i+1)},\bm{v}^{(i)}\right)\circledast \tilde{\bm{x}}_q^{(i+1)})  \right),\\
        &\nabla_{\bm{v}} g_{cc}(\bm{x}^{(i+1)},\bm{v}^{(i)}) \\
        &= 2 \sum_{q=1}^Q \sum_{\substack{q'=1\\ q'\neq q}}^Q \text{vec}\left( \textbf{a}(\theta_q)\left(\tilde{\textbf{r}}^{cc}_{q,q'}(\bm{x}^{(i+1)},\bm{v}^{(i)}) \circledast \tilde{\bm{x}}^{(i+1)}_{q'}\right)  \right),
    \end{aligned}
\end{equation}
where $\tilde{\bm{x}}_q^{(i+1)}=\textbf{a}^H(\theta_q)\textbf{X}^{(i+1)}$.

\color{black}
\bibliographystyle{IEEEtran}
\bibliography{IEEEabrv,references}

\end{document}